%% file: main.tex
\documentclass[10pt,journal,compsoc]{IEEEtran}
\ifCLASSOPTIONcompsoc
  \usepackage[nocompress]{cite}
\else
  \usepackage{cite}
\fi
\usepackage{algorithm}
\usepackage{algorithmicx}
\usepackage{algpseudocode}
\usepackage{graphicx}
\usepackage{textcomp}
\usepackage{xcolor}
\def\BibTeX{{\rm B\kern-.05em{\sc i\kern-.025em b}\kern-.08em
    T\kern-.1667em\lower.7ex\hbox{E}\kern-.125emX}}

\newcommand{\edit}[1]  {{{#1}}}

\usepackage{balance}
\usepackage{amsmath}
\usepackage{amsfonts,amssymb}
\usepackage{marvosym}
\usepackage{amsthm}

\usepackage{enumitem,kantlipsum}

\usepackage{ragged2e}
\usepackage{stmaryrd}

\newtheorem{theorem}{\textbf{Theorem}}
\newtheorem{definition}{\textbf{Definition}}
\newtheorem{proposition}{\textbf{Proposition}}

\newtheorem{lemma}{\textbf{Lemma}}

\begin{document}

\title{Differentially Private Numerical Vector Analyses\\ in the Local and Shuffle Model}




\author{
Shaowei~Wang,
Jin~Li,
Yuntong~Li,
Jin~Li,
Wei~Yang,
and Hongyang Yan

\IEEEcompsocitemizethanks{
\IEEEcompsocthanksitem Shaowei~Wang, Jin~Li, Yuntong~Li, Jin~Li, Hongyang Yan are with the Guangzhou University.\protect\\
E-mail: \{wangsw,lijin\}@gzhu.edu.cn
\IEEEcompsocthanksitem Wei~Yang is with the University of Science and Technology of China.
}
}



\IEEEtitleabstractindextext{%

\begin{abstract}
\justifying{Numerical vector aggregation plays a crucial role in privacy-sensitive applications, such as distributed gradient estimation in federated learning and statistical analysis of key-value data. In the context of local differential privacy, this study provides a tight minimax error bound of $O(\frac{ds}{n\epsilon^2})$, where $d$ represents the dimension of the numerical vector and $s$ denotes the number of non-zero entries. By converting the conditional/unconditional numerical mean estimation problem into a frequency estimation problem, we develop an optimal and efficient mechanism called Collision. In contrast, existing methods exhibit sub-optimal error rates of $O(\frac{d^2}{n\epsilon^2})$ or $O(\frac{ds^2}{n\epsilon^2})$. Specifically, for unconditional mean estimation, we leverage the negative correlation between two frequencies in each dimension and propose the CoCo mechanism, which further reduces estimation errors for mean values compared to Collision. Moreover, to surpass the error barrier in local privacy, we examine privacy amplification in the shuffle model for the proposed mechanisms and derive precisely tight amplification bounds. Our experiments validate and compare our mechanisms with existing approaches, demonstrating significant error reductions for frequency estimation and mean estimation on numerical vectors.}
\end{abstract}


}

\maketitle

\input{introduction.tex}
\input{preliminaries.tex}

\input{related.tex}
\input{bound.tex}

\input{mechanism.tex}
\input{shuflle.tex}
\input{mean.tex}

\input{experiments.tex}
\input{conclusion.tex}

\section*{Acknowledgements}
This work is extended from \cite{wang2021hiding} in the 30th International Joint Conference on Artificial Intelligence (IJCAI 2021). Shaowei Wang is supported by National Key Research and Development (R\&D) Program (Young Scientist Scheme No. 2022YFB3102400), National Natural Science Foundation of China (No.62102108), Natural Science Foundation of Guangdong Province of China (No.2022A1515010061), Guangzhou Basic and Applied Basic Research Foundation (No.202201010194, No.622191-098). This work is also supported by National Key Project of China (No. 2020YFB1005700), National Natural Science Foundation of China for Joint Fund Project (No. U1936218), and the Pazhou lab, Guangzhou, China. 


\bibliographystyle{IEEEtran}
\bibliography{main}

\clearpage
\input{appendix}


\end{document}

%% file: introduction.tex
\section{Introduction}
\label{sec:intro}   
With increasingly stringent data privacy regulations being enacted (e.g., the General Data Protection Regulation \cite{voigt2017eu} in the European Union, the California Consumer Privacy Act, and the Civil Code of the People's Republic of China), local differential privacy (LDP) has emerged as the \textit{de facto} standard for preserving data privacy in decentralizd settings. Stemming from the classical notion of differential privacy in the database community \cite{dwork2008differential}, LDP operates without trusting data aggregators or other third parties. It enables users/agents to sanitize their personal data locally (e.g., on mobile devices, or IoT sensors) and offers information-theoretically rigorous privacy protection. In comparison to cryptography-based privacy preservation approaches (e.g., homomorphic encryption \cite{naehrig2011can}, secure multi-party computation \cite{goldreich1998secure}), LDP is highly efficient and scalable for data aggregation involving millions or billions of users. Currently, many large internet service providers (such as Apple \cite{greenberg2016apple}, Google \cite{erlingsson2014rappor}, and Microsoft \cite{ding2017collecting}) are implementing LDP for regulatory compliance during user data collection and analysis. 

Additionally, to address the unacceptably high error barriers resulting from stringent LDP constraints, researchers have recently introduced the shuffle model \cite{bittau2017prochlo, erlingsson2019amplification} of differential privacy. In this model, messages from users are randomly permuted (by a shuffler, e.g., anonymous channels, trusted hardwares, and edge servers) before being sent to the aggregator/analyzer. This breaks the linkage between users and their messages, allowing messages to be concealed among others. Privacy is thus amplified after shuffling, enabling a lower local privacy level to satisfy a relatively higher privacy level (from the aggregator's perspective).

Numerical vectors are commonly found in user data for various applications, such as gradient estimation in federated learning \cite{konevcny2016federated,kairouz2021advances}, sensor readings \cite{cornelius2008anonysense}, and service usage histories \cite{mcsherry2009differentially} for user profile and usage analysis in web services. This study focuses on numerical vector analysis within the local and shuffle models of differential privacy. For clarity, we assume that the numerical vector $\mathbf{x}_i$ for user $i$ is a $d$-dimensional, $s$-sparse ternary vector \cite{wen2017terngrad,Ye2019PrivKVKD,sun2019conditional,gu2019pckv}, belonging to the set $\mathcal{X}^s$, defined as follows:
$$\mathcal{X}^s := \{\mathbf{x}\ |\ \mathbf{x} \in \{-1,0,1\}^d\ \text{and}\ \|\mathbf{x}\|_0=s\}.$$
This problem is pertinent to numerous real-world data aggregation tasks, including gradient estimation in federated learning and sensitive key-value data aggregation for user profile and usage analyses in web services.

\subsection{Federated Gradient Estimation}
Federated learning \cite{konevcny2016federated} investigates machine learning systems in distributed settings, enabling each party to maintain the privacy of their raw data. During each gradient descent iteration for training or updating a machine learning model, locally computed gradients $\mathbf{x}_i$ from participating parties (e.g., $n$ mobile users) are averaged by the federation server (e.g., a parameter server):
\begin{equation}\label{eq:nve}
\overline{\mathbf{x}} :=\frac{1}{n}\sum\nolimits_{i=1}^n \mathbf{x}_i.
\end{equation}
To enhance communication efficiency, local gradients are often discretized and sparsified \cite{wen2017terngrad,wangni2018gradient}.

The original work \cite{konevcny2016federated} considers sharing gradients to be more privacy-resistant than sharing raw data. However, recent studies show that the gradient $\mathbf{x}_i$ still poses privacy risks, as local raw data may be inferred with confidence from several transmitted gradients \cite{zhu19deep}. This highlights the need for rigorous privacy protection for local gradients.

\subsection{Key-value Data Aggregation}
We refer to key-value data as paired (key, value) mappings, where the key $j\in [d]$ represents an index, and the value $\mathbf{x}_j$ is numerical. Note that a value is considered $0$ if and only if the corresponding key is missing from or not defined in the key-value data; for any existing or defined keys, their corresponding values are binary as ${-1,1}$. For instance, a user might represent preferences for watched movies as key-value data, assigning a value of $1$ to movies they like and a value of $-1$ to movies they dislike.

Common analyses of key-value data involve estimating unconditional and conditional mean statistics. The unconditional mean statistic for the key $j$ is $\overline{\mathbf{x}}_j=\frac{1}{n}\sum_{i=1}^n \mathbf{x}_{i,j},$ the non-missing frequency of the key $j$ is: 
\begin{equation}\label{eq:frequency}
\underline{\mathbf{x}}_j :=\frac{1}{n}\#\{\mathbf{x}_{i,j}\ |\ \mathbf{x}_{i,j}\ for\ i\in[n]\ and\ \mathbf{x}_{i,j}\neq 0\},
\end{equation} 
and the conditional mean statistic is
$\overline{\mathbf{x}}_{\underline{j}} :=\overline{\mathbf{x}}_j/\underline{\mathbf{x}}_j.$

\subsection{Existing Results}
\label{subsec:existing}
Within the framework of $\epsilon$-LDP, theoretical minimax lower bounds for various statistical estimation problems have been established, including multinomial distribution estimation \cite{duchi2013local}, logistic regression/generalized linear model estimation \cite{duchi2018minimax}, and sparse covariance matrix estimation \cite{wang2019lower}. Specifically, \cite{duchi2018minimax} derives minimax lower bounds for multi-dimensional mean estimation in numerical vectors with bounded $\ell_1$-norm or $\ell_2$-norm. However, an $s$-sparse numerical vector is a special case of $\ell_1$-norm or $\ell_2$-norm bounded vector with identical absolute non-zero entries. It remains an open question whether it holds the same bounds as the general case or has tighter bounds. Recently, for a broad family of $\epsilon$-LDP estimation problems that can be cast as mean estimation problems, \cite{blasiok2019towards} studies sample complexity lower bounds under certain error tolerance $\alpha$, but their sample complexity results for $s$-sparse numerical vectors exhibit at least a $1/\alpha$ gap compared to our minimax optimal sample complexity results.

In practice, numerous $\epsilon$-LDP mechanisms have been proposed for statistical estimation, such as multinomial distribution estimation on categorical data \cite{duchi2013local,erlingsson2014rappor,kairouz2016discrete,wang2020locally} and one-dimensional mean estimation on numerical values \cite{wang2019collecting,sun2020bisample}. For $\epsilon$-LDP numerical vector or key-value data aggregation, existing approaches handle both dense numerical vectors (e.g., in \cite{nguyen2016collecting,duchi2018minimax,wang2022analyzing}) and sparse numerical vectors (e.g., in \cite{Ye2019PrivKVKD,sun2019conditional,gu2019pckv,zhou2022locally}). Specifically, \cite{Ye2019PrivKVKD,sun2019conditional} uniformly and randomly select one dimension from $[d]$ and transform the multi-dimensional estimation problem to a one-dimensional numerical/categorical problem. The work of \cite{gu2019pckv} follows a similar paradigm, but randomly selects one non-empty dimension from $s$ dimensions. However, as we will show in Section \ref{sec:related}, these mechanisms are sub-optimal.

To mitigate the high noise needed for LDP, \cite{erlingsson2019amplification,bittau2017prochlo} introduce a (semi-trusted) shuffler to hide private views in the crowd. The seminal work \cite{erlingsson2019amplification} shows that $n$ shuffled $\epsilon$-LDP views can preserve $(O({\epsilon}\sqrt{{\log(1/\delta)}/{n}}),\delta)$-differential privacy. The work \cite{cheu2019distributed} derives a similar conclusion specifically for binary randomized response messages. A later work \cite{balle2019privacy} considers private views from other users as a "privacy blanket" and derives tighter privacy amplification bounds $(O(\min\{\epsilon_0,1\}e^{\epsilon_0}\sqrt{\log(1/\delta)/{n}}),\delta)$. Recent works \cite{feldman2021hiding,feldman2023stronger} analyze the mixture property of arbitrary $\epsilon$-LDP randomizers and derive an asymptotically optimal bound of $(O((e^{\epsilon_0/2}-e^{-\epsilon_0/2})\sqrt{\log(1/\delta)/{n}}),\delta)$. This work shows that for specific $\epsilon$-LDP mechanisms, such as the proposed Collision and CoCo, it is possible to obtain tighter privacy amplification bounds.

\subsection{Our Contributions}
The contributions of this work are summarized as follows:
\begin{itemize}
\item{\textbf{Minimax lower bounds.\ } The squared error (or total variation error) lower bound of $\epsilon$-LDP $s$-sparse numerical vector mean estimation is $O(\frac{ds}{n\epsilon^2})$ (or $O(d\sqrt{\frac{s}{n\epsilon^2}})$). Our proof considers $s$-sparse numerical vectors that are decomposable, thus reducing the bounding procedure to the case of multiple multinomial distributions.}
\item{\textbf{An optimal mechanism via frequency estimation.\ } Since existing approaches are sub-optimal, we design a new mechanism: Collision, which matches the minimax lower bound. It has computational complexity $O(s)$ and communication complexity $O(\log s)$.}
\item{\textbf{An optimized mean estimation mechanism.\ } Exploiting the negative correlation between two frequencies for each dimension, we design an optimized mechanism, \emph{CoCo}, specifically for mean estimation, which further reduces estimation error by $15\%$.}
\item{\textbf{Tight privacy amplification in the shuffle model.\ } In the shuffle model of differential privacy, we derive exactly tight privacy amplification bounds for both Collision and CoCo. The amplification bounds are independent of dimension $d$ and sparsity parameter $s$, thus are favorable for high dimensional or even dense numerical vectors. When compared with existing results, our tight bounds save about $25\%$ privacy budget.}
\end{itemize}

The structure of the remaining paper is as follows: Section \ref{sec:prelim} provides background knowledge. Section \ref{sec:related} reviews the design of existing mechanisms and highlights their sub-optimality. In Section \ref{sec:bound}, we establish the minimax lower bounds.
Next, in Section \ref{sec:mechanism}, we propose the new mechanism, Collision, that matches the established lower bound.
In Section \ref{sec:shuffle}, we derive privacy amplification upper and lower bounds for the proposed mechanism in the shuffle model. Later, in Section \ref{sec:meanmechanism}, we propose an optimized mechanism for mean estimation. Section \ref{sec:experiments} presents experimental results. Finally, in Section \ref{sec:conclusion}, we conclude this work.

%% file: preliminaries.tex
\section{Preliminaries}\label{sec:prelim}
In this section, we introduce the definition of numerical vectors, differential privacy, and the minimax risks of private estimation. Commonly used notations are listed in Table \ref{tab:notations}.

\begin{table}
\renewcommand{\arraystretch}{1.17}
\caption{List of notations.}
\label{tab:notations}
\centering
\begin{tabular}{c|l}
\hline
\bfseries Notation & \bfseries Description\\
\hline
$[i]$ & $\{1,2,...,i\}$ \\
$[i:j]$ & $\{i,i+1,...,j\}$ \\
$\llbracket\ \  \rrbracket$ & Iverson bracket \\
\hline
$n$ & the number of users (data owners)\\
$d$ & the dimension of numerical vectors\\
$s$ & the sparsity parameter of numerical vectors\\
$\mathcal{X}^s$ & the domain of $s$-sparse numerical vector.\\
$\mathbf{x}_{i}$ & the data of user $i$\\

$\overline{\mathbf{x}}_{j}$ & the population mean of $j$-th dimension\\

$\underline{\mathbf{x}}_{j}$ & the non-missing frequency of $j$-th dimension\\

$\epsilon$ & the (local) privacy budget\\
$t$ & the outputting domain size\\
\hline
$\mathcal{D}$ & a distance measure over distributions\\
$\mathcal{S}$ & the shuffling algorithm in the shuffle model\\
$\epsilon_c$ & amplified privacy level in the shuffle model\\

\hline
\end{tabular}
\end{table}

\subsection{Numerical Vector}\label{subsec:nv}
We define a numerical vector $\mathbf{x}_i$ from every user $i$ as a $d$-dimensional, $s$-sparse ternary vector, with the domain defined as follows:
$$\mathcal{X}^s := \{\mathbf{x}\ |\ \mathbf{x} \in \{-1,0,1\}^d\ \text{and}\ \|\mathbf{x}\|_0=s\}.$$
Here, $s$ is the sparsity parameter: the number of non-zero elements in the numerical vector $\mathbf{x}_i$. Real-world real-valued numerical vectors can be transformed to $s$-sparse ternary vectors with limited precision loss, such as by max-min normalization and stochastic ternary discretization.

Additionally, we use the set form representation for the $s$-sparse vector. Let $j_-$ and $j_+$ denote events where the $j$-th element of $\mathbf{x}_i$ (i.e., $\mathbf{x}_{i,j}$) equals $-1$ and $1$, respectively. A numerical vector $\mathbf{x}$ can be represented in the set form as:
\small
$$\mathbf{Y}_{\mathbf{x}_{i}} :=\{j_-\ |\ j\in[d],\ \mathbf{x}_{i,j}=-1\} \bigcup \{j_+\ |\ j\in[d],\ \mathbf{x}_{i,j}=1 \}.$$
\normalsize

\subsection{Differential Privacy}\label{subsec:dp} 
One common tool for measuring distance between two distributions is hockey-stick divergence (see Definition \ref{def:hsd}), which satisfies data processing inequality \cite{sason2016f}.
\begin{definition}[Hockey-stick divergence]\label{def:hsd} The hockey-stick divergence between two random variables $P$ and $Q$ is:
\[\mathcal{D}_{e^{\epsilon}}(P \| Q) := \int\max\{0,P(x)-e^\epsilon Q(x)\} \mathsf{d}x,\]
where we use the notation $P$ and $Q$ to refer to both the random variables and their probability density functions.
\end{definition}

Two variables $P$ and $Q$ are $(\epsilon, \delta)$-indistinguishable if $\max{\mathcal{D}_{e^\epsilon}(P\|Q), \mathcal{D}_{e^\epsilon}(Q\| P)}\leq \delta$. For datasets $D$, $D'$ that are of the same size and differ only in one element, they are called \emph{neighboring datasets}. The definition of differential privacy with budget/level $(\epsilon,\delta)$ is as follows.

\begin{definition}[$(\epsilon,\delta)$-DP \cite{dwork2008differential}]\label{def:cdp}
Let $\mathcal{D}_K$ denote the output domain, a randomized mechanism $K$ satisfies $(\epsilon, \delta)$-differential privacy if, for any neighboring datasets $D, D'$, the $K(D)$ and $K(D')$ are $(\epsilon, \delta)$-indistinguishable.
\end{definition}

Let $K$ denote a randomized mechanism for sanitizing a single user's data. The differential privacy in the local model with privacy budget $\epsilon$ is as follows.

\begin{definition}[$\epsilon$-LDP \cite{duchi2013local}]\label{def:ldp}
Let $\mathcal{D}_K$ denote the output domain. A randomized mechanism $K$ satisfies local $\epsilon$-differential privacy if, for any data pair $\mathbf{x},\mathbf{x}' \in \mathcal{X}^s$, the $K(\mathbf{x})$ and $K(\mathbf{x}')$ are $(\epsilon, 0)$-indistinguishable.
\end{definition}

\subsubsection{The Shuffle Model of Differential Privacy}
In the shuffle model, a semi-trustable shuffler lies between the users and the data collector (e.g., the server/statistician) and uniform-randomly permutes randomized messages from users. We denote the randomization algorithm on the user side as $\mathcal{R}$ and the shuffling algorithm as $\mathcal{S}$. The privacy goal of the shuffle model is to ensure the shuffled messages $\mathcal{S}\circ\mathcal{R}(D)=\mathcal{S}(\mathcal{R}(x_1),...,\mathcal{R}(x_n))$ satisfy $(\epsilon_c,\delta)$-DP for all neighboring datasets:
\begin{definition}[$(\epsilon,\delta)$-DP in the shuffle model]\label{def:sdp}
A protocol $(\mathcal{R}, \mathcal{S})$ satisfies $(\epsilon, \delta)$-differential privacy in the shuffle model if, for any neighboring datasets $D, D'$, the $\mathcal{S}\circ\mathcal{R}(D)$ and $\mathcal{S}\circ\mathcal{R}(D')$ are $(\epsilon, \delta)$-indistinguishable.
\end{definition}

When the randomization algorithm $\mathcal{R}$ is an $\epsilon$-LDP mechanism, the seminal work shows $\mathcal{S}\circ\mathcal{R}$ actually preserves $({\epsilon}\sqrt{{144\log(1/\delta)}/{n}}, \delta)$-DP, which decreases with the number of users $n$. This phenomenon is known as \emph{privacy amplification via shuffling}. A very recent work \cite{feldman2023stronger} improves the bound to near-optimal $(O((e^{\epsilon/2}-e^{-\epsilon/2})\sqrt{\log(1/\delta)/{n}}),\delta)$. Specifically, when $\mathcal{R}$ satisfies several mixture properties, \cite{feldman2023stronger} shows the divergence between $\mathcal{S}\circ\mathcal{R}(D)$ and $\mathcal{S}\circ\mathcal{R}(D')$ is bounded by the divergence between a pair of two-dimension variables as in Theorem \ref{lemma:strongerclone}.

\begin{theorem}[Stronger clone reduction \cite{feldman2023stronger}]\label{lemma:strongerclone}
Given any $n+1$ inputs $x_1, x_1', x_2,..., x_n\in \mathcal{X}$, consider an algorithm $\mathcal{R}$ such that the output domain is finite and
\begin{align*}
&& \mathcal{R}(x_1)&=e^{\epsilon} \alpha \mathcal{Q}_1 + \alpha \mathcal{Q}_1'+(1-\alpha-e^{\epsilon} \alpha)\mathcal{Q}_1^*, &\\
&& \mathcal{R}(x_1')&=\alpha \mathcal{Q}_1 + e^{\epsilon} \alpha  \mathcal{Q}_1'+(1-\alpha-e^{\epsilon}\alpha)\mathcal{Q}_1^*, &\\
&& \forall i\in [2,n],\  \mathcal{R}(x_i)&=\alpha \mathcal{Q}_1 + \alpha \mathcal{Q}_1'+(1-2\alpha)\mathcal{Q}_i&
\end{align*}
holds for some $\epsilon\geq 0, \alpha\in [0,(e^{\epsilon}-1)/(e^{\epsilon}+1)]$ and some probability distributions $\mathcal{Q}_1,\mathcal{Q}_1',\mathcal{Q}_1^*,\mathcal{Q}_2,...,\mathcal{Q}_n$. Let $C\sim Binomial(n-1, 2\alpha)$, $A\sim Binomial(C, 1/2)$, and $\Delta_1=Bernoulli(e^{\epsilon} \alpha)$ and $\Delta_2=Bernoulli(1-\Delta_1, \alpha/(1-e^{\epsilon} \alpha))$; let $P_\alpha=(A+\Delta_1, C-A+\Delta_2)$ and $Q_\alpha=(A+\Delta_2, C-A+\Delta_1)$. Then for any distance measure $\mathcal{D}$ that satisfies the data processing inequality,
\begin{align*}
&\mathcal{D}(\mathcal{S}(\mathcal{R}(x_1),..,\mathcal{R}(x_n)) \| \mathcal{S}(\mathcal{R}(x_1'),..,\mathcal{R}(x_n)))
\leq \mathcal{D}(P_\alpha \| Q_\alpha). &
\end{align*}
\end{theorem}
For any $\epsilon_0$-LDP mechanism $\mathcal{R}$, it satisfies the mixture properties with parameters $\epsilon=\epsilon_0$ and $\alpha=(e^{\epsilon_0}-1)/(e^{\epsilon_0}+1)$. Furthermore, the distance $\mathcal{D}(P_\alpha \| Q_\alpha)$ increases with $\alpha$ when $\epsilon$ is fixed. Owing to the simplicity of the formulas for $P$ and $Q$, their hockey-stick divergence can be numerically computed in $\tilde{O}(n)$ time \cite{koskela2021tight} with a specified precision.

\subsection{Local Private Minimax Risks}
Assuming samples ${x_1, x_2, ..., x_n}$ are $n$ i.i.d. drawn from a distribution $P \in \mathcal{P}$. Let $\mathcal{K}_\epsilon$ denote the set of all possible mechanisms $\mathbf{K} = \{K_1, ..., K_n\}$ that each satisfies $\epsilon$-LDP. Taking the samples as input, a serial of (adaptive or non-adaptive) mechanisms $\mathbf{K} \in \mathcal{K}_\epsilon$ produce a list of sanitized views $\{z_1, z_2, ..., z_n\}$. If the parameter estimator: $$\widehat{\theta}=\widehat{\theta}(\{z_1,z_2,...,z_n\})$$ is derived from these private views while having no access to input samples $\{x_j\}_{j=1}^n$, the minimax MSE risk (under privacy budget $\epsilon$) is then:

\begin{equation*}
\begin{aligned}
&&&\mathfrak{M}_n(\theta(\mathcal{P}), \|\cdot\|_2^2,  \epsilon)&&\\
&&&:=\inf_{\mathbf{K}\in \mathcal{K}_\epsilon} \inf_{\widehat{\theta}}\ \sup_{P \in \mathcal{P}} \mathbb{E}_{P,\mathbf{K}}[\|\widehat{\theta}(z_1,z_2,...,z_n)-\theta(P)\|_2^2].&&\\
\end{aligned}
\end{equation*}

%% file: related.tex
\section{Closely Related Works}\label{sec:related}
Due to its broad applications, numerical vector aggregation with local and shuffle DP has been attracting increasing research attention. In addition to the literature reviewed in Section \ref{subsec:existing}, we focus here on the most closely related works from \cite{Ye2019PrivKVKD, sun2019conditional, gu2019pckv, ye2021privkvm, zhou2022locally, liu2020flame}.

\subsection{Numerical Vectors with Local DP}

Existing works on $\epsilon$-LDP numerical vector aggregation can mainly be categorized into two types: those that perform dimension sampling in a data-agnostic manner (e.g., PrivKV in \cite{Ye2019PrivKVKD, sun2019conditional}) and those that do so in a data-dependent manner (e.g., PCKV in \cite{gu2019pckv}).

\textbf{The PrivKV Mechanism \cite{Ye2019PrivKVKD}.} The seminal work by \cite{Ye2019PrivKVKD} on $\epsilon$-LDP key-value data suggests initially randomly sampling a dimension $j\in[d]$ from the key domain, followed by applying an $\epsilon$-LDP categorical mechanism to the corresponding (key, value) pair, which takes a value from ${(j,0),(j,1),(j,-1)}$. Here, $(j,0)$ indicates that the key is empty in the key-value data. In essence, the PrivKV mechanism is akin to dividing a population of $n$ into $d$ groups, where each group is used to estimate $\llbracket j_+\in \mathbf{Y}_{\mathbf{x}}\rrbracket$ and $\llbracket j_-\in \mathbf{Y}_{\mathbf{x}}\rrbracket$ for each $j\in[d]$ with a privacy budget of $\epsilon$. Given that the minimax lower error bound for estimating frequencies in a population of $n'$ with privacy budget $\epsilon$ and domain size $d'$ is $\Theta(\frac{d'}{n'\epsilon^2})$ \cite{duchi2018minimax}, the estimation error of $\llbracket j_+\in \mathbf{Y}_{\mathbf{x}}\rrbracket$ and $\llbracket j_-\in \mathbf{Y}_{\mathbf{x}}\rrbracket$ is $\Theta(\frac{d}{n\epsilon^2})$, since $n'=\frac{n}{d}$ and $d'=3$. Consequently, its total estimation error for frequencies or mean values of a $d$-dimensional vector is $O(\frac{d^2}{n\epsilon^2})$. This result exhibits a gap of ${d}/{s}$ from the optimal error rate in Theorem \ref{the:vectorminimax}. Analogous methodology and findings also apply to subsequent works in \cite{sun2019conditional,ye2021privkvm,liu2020flame}.

\textbf{The PCKV Mechanism \cite{gu2019pckv}.} The study by \cite{gu2019pckv} suggests sampling one key from the existing $s$ keys in key-value data. Subsequently, an $\epsilon$-LDP categorical mechanism is applied to the corresponding $1$-sparse numerical vector, which is equivalent to categorical data with a domain size of approximately $2d$. Considering that the minimax lower error bound for estimating frequencies in a population of $n'$ with privacy budget $\epsilon$ and domain size $d'$ is $\Theta(\frac{d'}{n'\epsilon^2})$, the total estimation error for scaled $\llbracket j_+\in \mathbf{Y}_{\mathbf{x}}\rrbracket$ and $\llbracket j_-\in \mathbf{Y}_{\mathbf{x}}\rrbracket$ in the PCKV mechanism is $\Theta(\frac{d}{n\epsilon^2})$, as $n'=n$ and $d'=2d$. Owing to the preceding sampling procedure, the scale factor is $s$, and the total variation error is amplified by $s^2$. Consequently, the total estimation error for $\llbracket j_+\in \mathbf{Y}_{\mathbf{x}}\rrbracket$ and $\llbracket j_-\in \mathbf{Y}_{\mathbf{x}}\rrbracket$ in the PCKV mechanism is $O(\frac{ds^2}{n\epsilon^2})$. This result presents a gap of $s$ from the optimal error rate in Theorem \ref{the:vectorminimax}.

\textbf{The Amplified PCKV-GRR Mechanism \cite{gu2019pckv}.}
In the PCKV mechanism, which employs the generalized randomized response (GRR \cite{wang2017locally}) as the base randomizer, privacy levels are enhanced through dimension sampling \cite{gu2019pckv}. Specifically, this mechanism can be applied with a privacy budget of $\epsilon'=\log(s(e^\epsilon-1)+1)$, where $\epsilon$ represents the original privacy budget. The mean squared estimation error for this case is given by $O(\frac{s e^\epsilon(s e^\epsilon+d-s-e^\epsilon)+(d-s)(s e^\epsilon+d-s-1)}{(e^\epsilon-1)^2})$, which equates to $O(\frac{d^2}{n\epsilon^2})$ when $\epsilon=O(1)$. It is important to note that the achieved estimation error exhibits a multiplicative gap of ${d}/{s}$ compared to the optimal error rate.

\textbf{The Succinct Mechanism \cite{zhou2022locally}.} Recently, \cite{zhou2022locally} proposes mapping pseudo-random ${+1,-1}$ values into a single bucket, clipping the bucket's summation to a norm of $\eta=O(\sqrt{s\log (n/\beta)})$, and adding Laplace noise with a scale of $2\eta/\epsilon$. The mean squared error of this approach is $O(\frac{ds\log n}{n\epsilon^2})$, resulting in a multiplicative gap of $\log n$ compared to the optimal rate. Moreover, the mechanism requires prior knowledge of the population size $n$, which may be impractical in certain scenarios (e.g., data collection in mobile/edge computing \cite{wang2023shuffle}). Although the mechanism is theoretically proven to be rate-optimal under the $\ell_{\infty}$ error (see Section \ref{subsec:maebound} for more details), its empirical $\ell_{\infty}$ errors lag behind our proposal by approximately $30\%$ in almost all settings (see Section \ref{subsec:meanexperiments}).

\subsection{Numerical Vectors in the Shuffle Model}
The shuffle model \cite{erlingsson2019amplification, cheu2019distributed, feldman2021hiding, feldman2023stronger,wang2023shuffle} and privacy amplification via shuffling has been successfully applied to numerical vectors, as demonstrated in recent studies (e.g., \cite{liu2020flame, sun2021ldp, girgis2021shuffled, scott2022aggregation}). Specifically, \cite{scott2022aggregation} independently sanitizes each dimension and transmits the sanitized vector to the shuffler; \cite{sun2021ldp} further proposes separately transmitting each dimension to the shuffler, breaking the linkage of $d$ dimensions for a single user. However, the local private mechanisms in \cite{sun2021ldp,scott2022aggregation} are sub-optimal due to budget splitting for each dimension. The work by \cite{girgis2021shuffled} addresses the sub-optimality issue through dimension sampling but does not exploit the sparsity in the gradient vector. The study \cite{liu2020flame} first selects $s$ significant dimensions from the gradient vector, then applies local private mechanisms and shuffle amplification. Nonetheless, the local privacy mechanism in \cite{liu2020flame} is also sub-optimal due to budget splitting for every selected dimension, and there is no rigorous privacy guarantee for selected dimensions. In contrast, our local randomizer ensures all messages are differentially private and is minimax optimal. Additionally, the privacy amplification bounds in this work are strictly tight.

%% file: bound.tex
\section{Minimax Lower Bounds}\label{sec:bound}


The Assouad's method \cite{yu1997assouad} is a widely used tool for lower bounding through multiple hypothesis testing. It defines a hypercube $\mathcal{V}=\{-1,1\}^d$ ($d\in \mathbb{N}^+$) and a family of distributions $\{P_\nu\}_{\nu\in \mathcal{V}}$ indexed by the hypercube. A distribution family is said to induce a $2\tau$-Hamming separation for the loss $\|\cdot\|_2^2$ if a vertex mapping (a function $\kappa:\theta(\mathcal{P}) \mapsto \{-1,1\}^d$) exists, satisfying:
$$\|\theta-\theta(P_\nu)\|_2^2 \geq 2\tau\sum_{j=1}^d \llbracket[\kappa(\theta)]_j\neq \nu_j\rrbracket.$$
Assuming that nature first uniformly selects a vector $V\in \mathcal{V}$, and the samples ${\mathbf{x}_1,...,\mathbf{x}_n}$ are drawn from the distribution $P_\nu$ with $V=\nu$, these samples are then used as input for $\epsilon$-LDP mechanisms $\mathbf{K}$. The literature \cite{duchi2018minimax} presents an $\epsilon$-LDP version of Assouad's method, as follows.

\begin{lemma}[Private Assouad bound \cite{duchi2018minimax}]\label{lemma:ldpassouad}
Let $P_{+j}=\frac{1}{2^{d-1}}\sum_{\nu:\nu_j=1} P_\nu$ and $P_{-j}=\frac{1}{2^{d-1}}\sum_{\nu:\nu_j=-1} P_\nu$, we have
$$\mathfrak{M}_n(\theta(\mathcal{P}), \|\cdot\|_2^2)\geq d\cdot\tau[1-(\frac{n(e^\epsilon-1)^2}{2d} F_{\mathbb{B}_\infty(\mathcal{X}^s),\mathcal{P}})^{\frac{1}{2}}],$$
where $\mathbb{B}_\infty(\mathcal{X}^s)$ denote the collection of function $\gamma$ with supremum norm bounded by $1$ as:
$$\mathbb{B}_\infty(\mathcal{X}^s):=\{\gamma:\mathcal{X}^s \mapsto \mathbb{R}\ \ |\ \ \|\gamma\|_\infty \leq 1\},$$
and maximum possible discrepancy $F_{\mathbb{B}_\infty(\mathcal{X}^s),\mathcal{P}}$ is defined as:
$$\sup_{\gamma\in \mathbb{B}_\infty(\mathcal{X}^s) }\sum_{i=1}^d\big(\int_{\mathcal{X}^s}\gamma(x)(\textsf{d} P_{+j}(x)-\textsf{d} P_{-j}(x))\big)^2.$$
\end{lemma}

We consider numerical vectors that can be decomposed into $s$ buckets, with each bucket containing $\frac{d}{s}$ indices and only one non-zero entry. We then define a hypercube of length $d$ and construct a class of $\frac{2\delta^2 s^2}{d^2}$-Hamming separated probability distributions. Following Lemma \ref{lemma:ldpassouad}, we bound the maximum possible marginal distance $F_{\mathbb{B}_\infty(\mathcal{X}^s),\mathcal{P}}$ under the value of $\frac{8\delta^2 s}{d}$. Theorem \ref{the:vectorminimax} provides the final lower bounds for the problem of local private numerical vector mean estimation.


\begin{theorem}\label{the:vectorminimax}
For the numerical vector aggregation problem, for any $\epsilon$-LDP mechanism, there exists a universal constant $c> 0$ such that for all $\epsilon\in (0,1]$,
$$\mathfrak{M}_n(\theta(\mathcal{P}), \|\cdot\|_2^2, \epsilon)\geq c\cdot\min\{\frac{s^2}{d},\frac{ds}{n\epsilon^2}\}.$$
\end{theorem}
\begin{proof}
See Appendix \ref{app:minimax}.
\end{proof}

To understand the minimax rate, we can consider the non-private error rate of decomposable numerical vector aggregation, which is $
\mathbb{E}[\|\widehat{\theta}-\theta\|_2^2]\leq\sum_{i=1}^d \mathbb{E}[\|\widehat{\theta}_i-\theta_i\|_2^2]\leq \frac{4s}{n}.$
Thus, the enforcement of local $\epsilon$-LDP causes the effective sample size to decrease from $n$ to $O(n\epsilon^2/d)$.

Now consider the $\ell_1$-norm error metric, the estimation error lower bounds can be derived as $O(\frac{d\sqrt{s}}{\sqrt{n\epsilon^2}})$ (see Theorem \ref{the:vectorminimax1}). Compared to the non-private error rate for decomposable numerical vector data:
\begin{equation*}\label{eq:nonprivaterisk1}
\mathbb{E}[\|\widehat{\theta}-\theta\|_1]\leq\sum_{a=1}^s\sum_{j=1}^{d/s} \mathbb{E}[|\widehat{\theta}_{a,j}-\theta_{a,j}|]\leq 2 s\sqrt{\frac{d/s}{n}},
\end{equation*}
this also demonstrates that the $\epsilon$-LDP reduces the effective sample size from $n$ to $O(n \epsilon^2/d)$.

\begin{theorem}\label{the:vectorminimax1}
For the numerical vector aggregation problem, for any $\epsilon$-LDP mechanism, there exists a universal constant $c> 0$ such that for all $\epsilon\in (0,1]$,
$$\mathfrak{M}_n(\theta(\mathcal{P}), \|\cdot\|_1, \epsilon)\geq c\cdot\min\{\frac{s}{2},\frac{d\sqrt{s}}{\epsilon^2 \sqrt{n}}\}.$$
\end{theorem}
\begin{proof}
The proof for the $\|\cdot\|_1$ error follows a similar procedure to the one for the $\|\cdot\|_2^2$ error, with some differences in multiplicative factors in steps 2 and 4.
In step 2, Equation (\ref{eq:separation}) now becomes:
\begin{equation*}
    \|\widehat{\theta}-\theta_\nu\|_1\geq \frac{\delta s}{d}\sum_{j=1}^{l}\sum_{a=1}^s\llbracket\widehat{\nu}_{a_j}\neq \nu_{a_j}\rrbracket.
    \end{equation*}
Consequently, the Hamming separation parameter with respect to the $\ell_1$-norm is $\frac{\delta s}{d}$. Later, in step 4, according to Lemma \ref{lemma:ldpassouad}, we obtain:
$$\max_{\nu \in \mathcal{V}} \mathbb{E}_{P_\nu}[\|\widehat{\theta}-\theta_\nu\|_1]\geq {\delta s}[1-(4n(e^\epsilon-1)^2\delta^2s/d^2)^\frac{1}{2}].$$
By choosing the parameter $\delta^2$ at $\min\{1, d^2/(16n s (e^\epsilon-1)^2)\}$, we establish the lower bound as:
$$\mathfrak{M}_n(\theta(\mathcal{P}), \|\cdot\|_1, \epsilon)\geq \min\{\frac{s}{2},\frac{d \sqrt{s}}{8(e^\epsilon-1)\sqrt{n}}\}.$$
\end{proof}

%% file: mechanism.tex
\section{Optimal Frequency Mechanism}\label{sec:mechanism}
In this section, we propose a frequency-based mechanism (i.e., Collision) for $\epsilon$-LDP numerical vector aggregation that matches minimax error lower bounds.

To mitigate the curse of dimension/density on the performance of numerical vector estimation, existing $\epsilon$-LDP mechanisms employ the paradigm of \textit{dimension/key sampling \& categorical randomization}, which, however, fails to achieve the optimal statistical rate. We propose to first condense the numerical vector to prevent interference from the original dimension, and then sample one element from the dense vector using the exponential mechanism \cite{mcsherry2007mechanism} to avoid splitting the privacy budget (in order to prevent dependence on $s^2$). We define an element domain: $$\mathcal{Y}=\{1_-,1_+,2_-,2_+,...,d_-,d_+\},$$ and represent an input $\mathbf{Y}——\mathbf{x}$ as a subset of $\mathcal{Y}$ with size $s$. We also define an output domain as $\mathcal{Z}=\{1,2,...,t\}$. The Collision mechanism probabilistically outputs one item $z \in \mathcal{Z}$, the probability of which corresponds to whether the item has a collision with hashed events in $\mathbf{Y}_{\mathbf{x}}$. Here, the hash function $H: \mathcal{Y}\mapsto \mathcal{Z}$ is uniformly chosen at random from a finite domain $\mathcal{H}$ by each user independently, with an identical (and often uniform) distribution $P_\mathcal{H}:\mathcal{H}\mapsto [0,1]$. We present the design of the Collision mechanism in Definition \ref{def:design}.


\begin{definition}[$(d,s,\epsilon,t)$-Collision Mechanism]\label{def:design}
Given a random-chosen hash function $H: \mathcal{Y}\mapsto \mathcal{Z}$ according to distribution $P_\mathcal{H}$,  \edit{taking} an $s$-sparse numerical vector $\mathbf{Y}_\mathbf{x} \subseteq \mathcal{Y}$ as input, the Collision mechanism randomly outputs an element $z \in \mathcal{Z}$ according to following probability design:
\begin{equation*}\label{eq:distance}
    \mathbb{P}[z | \mathbf{x}]= \left\{
    \begin{array}{@{}lr@{}}
        \frac{e^\epsilon}{\Omega},\ \ \ \ \ \ \ \ \ \ \ \ \ \ \ \ \ \ \ \ \ \ \ \ \text{if } \exists\ y \in \mathbf{Y}_\mathbf{x} \text{ that } z=H(y);\\
        \frac{\Omega-e^\epsilon\cdot\#\{H(y)\ |\ H(y)\ for\ y\ \in \mathbf{Y}_\mathbf{x}\}}{(t-\#\{H(y)\ |\ H(y)\ for\ y\ \in \mathbf{Y}_\mathbf{x}\})\cdot \Omega}.\ \ \ \ \ \ otherwise.\\
    \end{array}
    \right.
\end{equation*}
The normalization factor $\Omega=s\cdot e^\epsilon+t-s$. An unbiased estimator of indicator $\llbracket j_b\in \mathbf{Y}_{\mathbf{x}}\rrbracket $ for $b\in\{-1,1\}$ and $j\in[d]$ is:
$$\widehat{\llbracket j_b\in \mathbf{Y}_{\mathbf{x}}\rrbracket }=\frac{\llbracket H(j_b)= z\rrbracket -1/t}{{e^\epsilon}/{\Omega}-1/t}.$$
\end{definition}

Figure \ref{fig:collision} (b) and (a) demonstrate the probability design of the Collision mechanism on numerical vector $[0,0,1,0,-1,0]$ when hash values conflict with each other or not, respectively. It can be seen as an $s$-item generalization of the prevalent local hash \cite{wang2017locally} for $1$-item categorical data. When $s\geq 2$, the hashed values may coincide with each other; thus, simply restraining the proportional probability in $\{1, e^{\epsilon}\}$ as \cite{wang2017locally} will cause inconsistency in the normalization factor $\Omega'=\#\{H(y)\ |\ H(y)\ for\ y\ \in \mathbf{Y}_\mathbf{x}\}\cdot (e^\epsilon-1)+t$ for different input $\mathbf{Y}_\mathbf{x}$, and hence violate $\epsilon$-LDP. Therefore, we fix $\Omega$ at $s\cdot e^\epsilon+t-s$ and uniformly redistribute the probability of coincided hash values to the remaining output domain as $\frac{\Omega-e^\epsilon\cdot\#\{H(y)\ |\ H(y)\ for\ y\ \in \mathbf{Y}_\mathbf{x}\}}{(t-\#\{H(y)\ |\ H(y)\ for\ y\ \in \mathbf{Y}_\mathbf{x}\})\cdot \Omega}$. In the Collision mechanism, the proportional probability is relaxed to $[1, e^{\epsilon}]$. We note that sampling one item from $\mathbf{Y}_\mathbf{x}$ and then feeding it into local hash \cite{wang2017locally} leads to $O(\frac{d s^2}{n \epsilon^2})$ squared error (see a similar analysis in Section \ref{subsec:existing} for PCKV \cite{gu2019pckv}).

\begin{figure}
\vspace*{-1em}
\begin{center}
\centerline{\includegraphics[width=83mm]{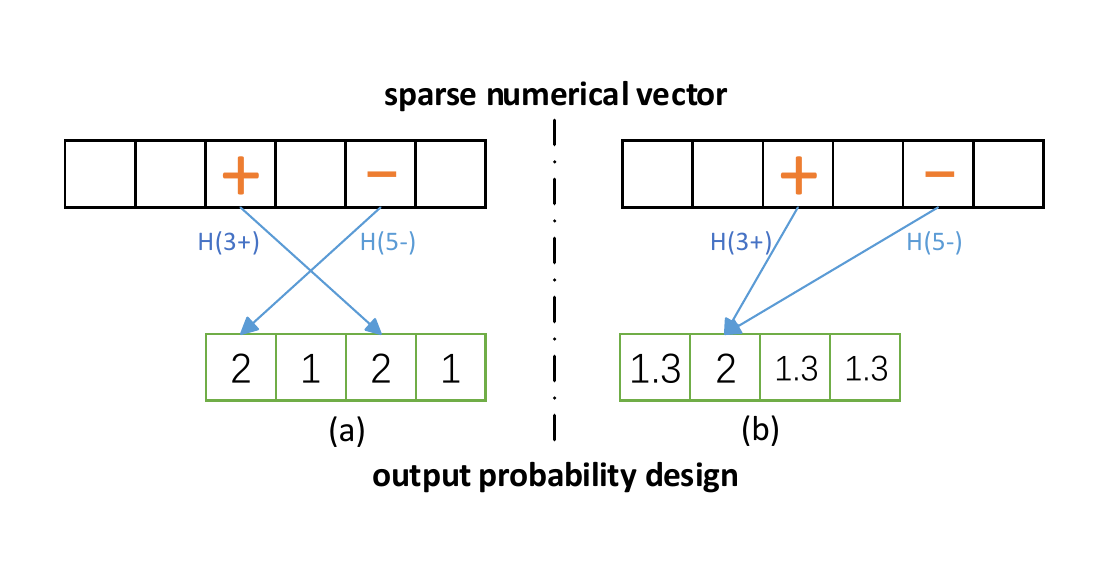}}
\vspace*{-2.0em}
\caption{An illustration of the Collision mechanism without hash conflict (a) and with hash conflicts (b), where $d=6$, $s=2$, $t=4$ and $\epsilon=\log(2)$.}
\label{fig:collision}
\end{center}
\vspace*{-0.0em}
\end{figure}

The local privacy guarantee of the mechanism is provided in Proposition \ref{the:ldp}, which is evident since $s\geq \#{H(y)\ |\ H(y)\ for\ y\ \in \mathbf{Y}_\mathbf{x}}$. The utility-optimality guarantee of the mechanism is presented in Theorem \ref{the:upperbound}. For $\epsilon=O(1)$, its computational complexity is bounded by $t^*\approx s+2s-1+s\cdot e^\epsilon=O(s)$, and communication complexity is $\log_2(2s-1+s\cdot e^\epsilon)=O(\log s)$.

\begin{proposition}\label{the:ldp}
The $(d,s,\epsilon,t)$-Collision mechanism in Definition \ref{def:design} satisfies $\epsilon$-LDP for numerical vector data.
\end{proposition}

\begin{theorem}\label{the:upperbound}
Given privacy budget $\epsilon=O(1)$, with optimal choice of the output parameter $t^*$, the mean estimation error of $(d,s,\epsilon,t)$-Collision mechanism for numerical vector is $O(\frac{ds}{n\epsilon^2})$.
\end{theorem}
\begin{proof}
Recall that the $j$-th mean value $\overline{\mathbf{x}}_j=\frac{1}{n}\sum_{i=1}^n \mathbf{x}_{i,j}$ equals to $\frac{1}{n}\sum_{i=1}^n (\llbracket j_{+}\in \mathbf{Y}_{\mathbf{x}}\rrbracket -\llbracket j_{-}\in \mathbf{Y}_{\mathbf{x}}\rrbracket )$. Assuming the uniform randomness and independence of hash functions in $\mathcal{H}$, we have each observed indicator $\llbracket H(j_b)=z\rrbracket$ as a Bernoulli random variable with a success rate $\frac{e^\epsilon}{\Omega}$ (when $j_b \in \mathbf{Y}_{\mathbf{x}}$) or a success rate $\frac{1}{t}$ (when $j_b \notin \mathbf{Y}_{\mathbf{x}}$). Consequently, the mean squared error of estimated frequencies is:
\begin{equation*}
\begin{aligned}
&&&{Var}[\widehat{\overline{\mathbf{x}}}]\leq 2\sum_{j=1}^d\sum_{b\in[-1,1]}{Var}\big[\widehat{\llbracket j_b\in \mathbf{Y}_{\mathbf{x}}\rrbracket}\big]&\\
&&&\leq\frac{2}{n}\cdot{\frac{s\cdot {e^\epsilon}/{\Omega}(1-{e^\epsilon}/{\Omega})+(2d-s)\cdot 1/t(1-1/t)}{({e^\epsilon}/{\Omega}-1/t)^2}}.
\end{aligned}
\end{equation*}
Taking the previous formula as a function of continuous $t$, the function is indeed convex when $d\geq t\geq s$. Choosing an approximate optimal $t^*$ at around ${2s-1+s\cdot e^\epsilon}$, we obtain:
\begin{equation*}
\begin{aligned}
&&{Var}[\widehat{\overline{\mathbf{x}}}]&\leq \frac{2d\cdot\Theta(s^3)+\epsilon\cdot\Theta(s^3)}{n\cdot\epsilon^2\cdot(-1 + (2 + \epsilon)\cdot s)^2}\leq O(\frac{ds}{n\epsilon^2}).&&\\
\end{aligned}
\end{equation*}
We note that a similar conclusion applies to non-missing frequency estimation (refer to the beginning of Section \ref{sec:meanmechanism}).
\end{proof}

%% file: shuflle.tex
\section{Privacy Amplification in Shuffle Model}\label{sec:shuffle}
When a semi-trusted shuffler is positioned between users and the aggregator, the aggregator only observes the shuffled private views $\mathcal{S}(z_1,z_2,...z_n)$, thereby amplifying the privacy level. This section aims to analyze the privacy amplification upper and lower bounds of $n$ shuffled private views from the Collision mechanism.

\subsection{Amplification Upper Bounds}
To prove the privacy amplification upper bounds based on Lemma \ref{lemma:strongerclone}, we begin by analyzing the mixture properties of the Collision mechanism. Given the hash function space $\mathcal{H}$ and distribution $P_{\mathcal{H}}$, with $\mathbb{P}[H]$ denoting the probability $P_{\mathcal{H}}[H]$ of selecting $H$, we demonstrate in the following lemma that the $(d,s,\epsilon,t)$-Collision mechanism has mixture parameter $\beta$.

\begin{lemma}[Mixture Properties]\label{lemma:mixture}

Let $x_1$, $x_1'$, $x_2$, $\ldots$, and $x_n$ be elements of the set $\mathcal{X}^s$. Let $H(\mathbf{Y}_{x_1})$ denote the set of hashed values of $\mathbf{Y}_{x_1}$, i.e., $\{H(y)\ |\ H(y)\ for\ y\ \in \mathbf{Y}_{x_1}\}$, and let $\mathcal{R}$ denote the $(d,s,\epsilon,t)$-Collision mechanism with $t>s$. We demonstrate the existence of distributions $\mathcal{Q}_1$, $\mathcal{Q}_1'$, $\mathcal{Q}_1^*$, $\mathcal{Q}_2$, $\ldots$, and $\mathcal{Q}_n$ that satisfy the following properties:
\begin{align}
\label{eq:mix1}&& \mathcal{R}(x_1^0)&=e^{\epsilon} \beta \mathcal{Q}_1 + \beta \mathcal{Q}_1'+(1-\beta-e^{\epsilon} \beta)\mathcal{Q}_1^* &\\
\label{eq:mix2}&& \mathcal{R}(x_1^1)&=\beta \mathcal{Q}_1 + e^{\epsilon} \beta  \mathcal{Q}_1'+(1-\beta-e^{\epsilon}\beta)\mathcal{Q}_1^* &\\
\label{eq:mix3}&& \forall i\in [2:n],\  \mathcal{R}(x_i)&=\beta \mathcal{Q}_1^0 + \beta \mathcal{Q}_1^1+(1-2\beta)\mathcal{Q}_i&
\end{align}
where $\beta=\sum_{H\in \mathcal{H}}\mathbb{P}[H]\cdot \frac{(e^\epsilon-1)(s-|H(\mathbf{Y}_{x_1})\bigcap H(\mathbf{Y}_{x_1'})|)}{s e^\epsilon+t-s}$.
\end{lemma}
\begin{proof}
See Appendix \ref{app:mixture}.
\end{proof}

Given Lemma \ref{lemma:mixture}, then combining the reduction in Lemma \ref{lemma:strongerclone}, the monotonic property of $\mathcal{D}(P_{\beta}\|Q_{\beta})$ \cite[Lemma 5.1]{feldman2023stronger}, and $\beta\leq \frac{s(e^\epsilon-1)}{s e^\epsilon+t-s}$ (equality holds when $H(\mathbf{Y}_{x_1})\bigcap H(\mathbf{Y}_{x_1'})=\Phi$), we arrive at the main theorem for privacy amplification upper bounds (in Theorem \ref{the:upper}).

\begin{figure*}[!htb]
\begin{center}
\centerline{\includegraphics[width=175mm]{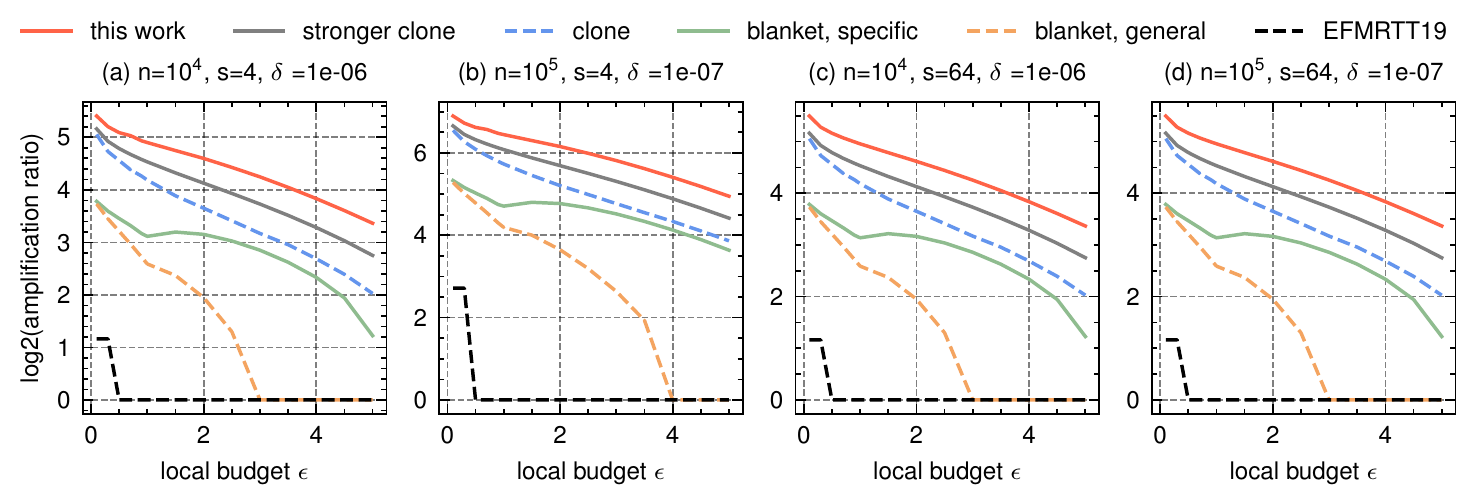}}
\vskip -0.15in
\caption{Comparison of amplification effects (base $2$ logarithm of amplification ratio $\frac{\epsilon}{\epsilon_c}$, the higher the better, where $\epsilon_c$ is the amplified privacy level in various amplification approaches) of Collision mechanism with $n=10^4$ or $10^5$, sparsity parameter $s=4$ or $64$, and varying local budget $\epsilon\in [0.1, 5.0]$. The hyperparameter $t$ is set to $t^*=\lfloor s e^\epsilon+2s-1\rfloor$ in the Collision mechanism.}
\label{fig:shuffle}
\end{center}
\vspace*{-1.5em}
\end{figure*}

\begin{theorem}[Amplification Upper Bounds]\label{the:upper}
Let $\mathcal{R}$ denote the $(d,s,\epsilon,t)$-Collision mechanism (assumed $t>s$), and let $\alpha=\frac{s(e^\epsilon-1)}{s e^\epsilon+t-s}$, then for any neighboring datasets $D,D'$, we have:
\begin{equation}\label{eq:upper}
\mathcal{D}(\mathcal{S}\circ\mathcal{R}(D)\| \mathcal{S}\circ\mathcal{R}(D'))\leq \mathcal{D}(P_{\alpha}\|Q_{\alpha}).
\end{equation}
\end{theorem}
\begin{proof}
Without loss of generality, we will consider two neighboring datasets, denoted as $D$ and $D'$, that differ only in the first datum. That is, $D=\{x_1,x_2,...,x_n\}$ and $D'=\{x_1',x_2,...,x_n\}$. Let $\beta$ be defined as $\sum_{H\in \mathcal{H}}\mathbb{P}[H]\cdot \frac{(e^\epsilon-1)(s-|H(\mathbf{Y}_{x_1})\bigcap H(\mathbf{Y}_{x_1'})|)}{s e^\epsilon+t-s}$. By invoking Lemma \ref{lemma:mixture} and Lemma \ref{lemma:strongerclone}, we obtain the following results:
$$\mathcal{D}(\mathcal{S}\circ\mathcal{R}(D)\| \mathcal{S}\circ\mathcal{R}(D'))\leq \mathcal{D}(P_{\beta}\|Q_{\beta}).$$

In addition, under the fixed value of $e^\epsilon$, the data processing inequality of distance measure $\mathcal{D}$ leads to the monotonically non-decreasing property of $\mathcal{D}(P_{\beta}\|Q_{\beta})$ with respect to $\beta$ \cite[Lemma 5.1]{feldman2023stronger}. By taking into account the inequality $\beta\leq \frac{s(e^\epsilon-1)}{s e^\epsilon+t-s}\leq \alpha$, we are able to draw the final conclusion.
\end{proof}

As the indistinguishable level between $\mathcal{S}\circ\mathcal{R}(D)$ and $\mathcal{S}\circ\mathcal{R}(D')$ is upper bounded by the indistinguishable level between $P_{\frac{s(e^\epsilon-1)}{s e^\epsilon+t-s}}$ and $Q_{\frac{s(e^\epsilon-1)}{s e^\epsilon+t-s}}$, we now focus on deriving the indistinguishable level of the latter pair. It is common in practice that $\delta\in (0,1]$ is fixed (e.g., $\delta=O(1/n)$), and one wants to know the minimum $\epsilon_c$ such that $P_{\frac{s(e^\epsilon-1)}{s e^\epsilon+t-s}}$ and $Q_{\frac{s(e^\epsilon-1)}{s e^\epsilon+t-s}}$ are $(\epsilon_c,\delta)$-indistinguishable. Directly solving the optimization problem is intractable; however, when $\epsilon_c$ is fixed, one can easily numerically compute $\mathcal{D}_{e^{\epsilon_c}}(P_{\frac{s(e^\epsilon-1)}{s e^\epsilon+t-s}}|Q_{\frac{s(e^\epsilon-1)}{s e^\epsilon+t-s}})$ and $\mathcal{D}_{e^{\epsilon_c}}(Q{\frac{s(e^\epsilon-1)}{s e^\epsilon+t-s}}|P_{\frac{s(e^\epsilon-1)}{s e^\epsilon+t-s}})$ (see reference \cite{koskela2021tight} for an $\tilde{O}(n)$ implementation). Finally, use the fact that the above hockey-stick divergence is monotonic w.r.t. $\epsilon_c\in [0, \epsilon]$, one can solve the minimization problem with satisfactory precision via binary search (e.g., in \cite{feldman2021hiding}).

 We compare our amplification upper bounds based on Theorem \ref{the:upper} with known bounds in the literature, including the closed-form amplification bound in \cite{erlingsson2019amplification} (denoted as \textit{EFMRTT19}), numerical bounds by privacy blanket \cite{balle2019privacy} (with both general parameter $1-e^{-\epsilon}$ and specific parameter $\frac{t}{s e^\epsilon+t-s}$ on total variation similarity), the numerical clone reduction \cite{feldman2021hiding}, and the numerical stronger clone reduction \cite{feldman2023stronger}. Some representative results are presented in Figure \ref{fig:shuffle}, which implies our bounds are tighter and save about $20\%$-$30\%$ privacy budget. This also indicates that existing bounds still overestimate privacy consumption, while our bounds match amplification lower bounds when local budget $\epsilon> \log(1+1/s)$ (see the next subsection).

\subsection{Amplification Lower Bounds}
In this section, we show that the amplification upper bounds in the former subsection are actually tight. Specifically, we provide worst-case scenarios where the quantity $\mathcal{D}(\mathcal{S}\circ\mathcal{R}(D)\| \mathcal{S}\circ\mathcal{R}(D'))$ is lower bounded by $\mathcal{D}(P_{\frac{s(e^\epsilon-1)}{s e^\epsilon+t-s}}\|Q_{\frac{s(e^\epsilon-1)}{s e^\epsilon+t-s}})$. The underlying idea is separately counting the observed elements $z$ in $[t]$ based on whether $z \in H(\mathbf{Y}_{x_1})$ or $z \in H(\mathbf{Y}_{x'_1})$, and subsequently summarizing them as Binomial counts.

\begin{theorem}[Amplification Lower Bounds]\label{the:lower}
Let $\mathcal{R}$ denote the $(d,s,\epsilon,t)$-Collision mechanism (assumed $t\geq 3s$), then there exists $\mathcal{H}$ and neighboring datasets $D,D'$ such that:
$$\mathcal{D}(\mathcal{S}\circ\mathcal{R}(D)\| \mathcal{S}\circ\mathcal{R}(D'))\geq \mathcal{D}(P_{\frac{s(e^\epsilon-1)}{s e^\epsilon+t-s}}\|Q_{\frac{s(e^\epsilon-1)}{s e^\epsilon+t-s}}).$$
\end{theorem}
\begin{proof}
Considering hash functions $\mathcal{H}$ and two neighboring datasets $D=\{x_1,x_2=x^*,..., x_n=x^*\}$ and $D'=\{x_1',x_2=x^*,..., x_n=x^*\}$ such that for any $H\in \mathcal{H}$, all three following equations holds (achievable when $t\geq 3s$):
\begin{align*}
H(\mathbf{Y}_{x_1})\bigcap H(\mathbf{Y}_{x_1'})=\Phi,\\
H(\mathbf{Y}_{x_1})\bigcap H(\mathbf{Y}_{x^*})=\Phi,\\
H(\mathbf{Y}_{x_1'})\bigcap H(\mathbf{Y}_{x^*})=\Phi.
\end{align*}

Now consider shuffled messages $\mathcal{S}(\mathcal{R}(x_1),...,\mathcal{R}(x^*))$ and $\mathcal{S}(\mathcal{R}(x_1'),...,\mathcal{R}(x^*))$. We define a post-processing function $g:\mathcal{H}\times\mathcal{Z}\mapsto \mathbb{N}^2$ on each message as follows (for any output $H,z\in \mathcal{H}\times\mathcal{Z}$):
\begin{equation*}
    g(H,z) :=\left\{
    \begin{array}{@{}lr@{}}
        (1, 0), & \text{if } z\in H(\mathbf{Y}_{x_1});\\
        (0, 1), & \text{if } z\in H(\mathbf{Y}_{x_1'});\\
        (0, 0), & \text{else.}
    \end{array}
    \right. 
\end{equation*}
Let us define a function $g_n:(\mathcal{H}\times\mathcal{Z})^n\mapsto \mathbb{N}^2$, which maps a set of $n$ shuffled messages $S$ to the summation of $g(s)$ for all $s\in \mathcal{S}$. It can be observed that $g_n({\mathcal{R}(x_1),\mathcal{R}(x^*),...,\mathcal{R}(x^*)})\stackrel{d}{=} P_{\frac{s(e^\epsilon-1)}{s e^\epsilon+t-s}}$ and $g_n(\{\mathcal{R}(x_1),\mathcal{R}(x^),...,\mathcal{R}(x^*)\})\stackrel{d}{=} Q_{\frac{s(e^\epsilon-1)}{s e^\epsilon+t-s}}$. Here, the notation $\stackrel{d}{=}$ denotes that the two random variables have the same distribution. Finally, we use the data processing inequality of Hockey-stick divergence (or any other distance measure $\mathcal{D}$ satisfying data processing inequality) to arrive at the conclusion.
\end{proof}

We present the amplification lower bound in Theorem \ref{the:lower}. Since the upper bound in Theorem \ref{the:upper} matches the lower bound, we conclude that the privacy amplification results in the former subsection are precisely tight (when $t\geq 3s$).

%% file: mean.tex
\section{Optimized Mean Mechanism}\label{sec:meanmechanism}
Previous sections mainly consider frequency estimation over the event domain $\mathcal{Y}=\{1_-,1_+,2_-,2_+,...,d_-,d_+\}$, which acts as intermediate results for both mean estimation and conditional mean estimation of numerical vectors. Specifically, the $j$-th mean value $\overline{\mathbf{x}}_j=\frac{1}{n}\sum_{i=1}^n \mathbf{x}_{i,j}$ equals to $\frac{1}{n}\sum_{i=1}^n (\llbracket j_{+}\in \mathbf{Y}_{\mathbf{x}}\rrbracket -\llbracket j_{-}\in \mathbf{Y}_{\mathbf{x}}\rrbracket )$, the $j$-th non-missing frequency $\underline{\mathbf{x}}_j=\frac{\#\{\mathbf{x}_{i,j}\ |\ \mathbf{x}_{i,j}\ for\ i\in[n]\ and\ \mathbf{x}_{i,j}\neq 0\}}{n}$ equals to $\frac{1}{n}\sum_{i=1}^n (\llbracket j_{+}\in \mathbf{Y}_{\mathbf{x}}\rrbracket +\llbracket j_{-}\in \mathbf{Y}_{\mathbf{x}}\rrbracket )$. According to the variance bounds of the plus/minus of two random variables, we have: 
\small{$$Var\big[\llbracket\widehat{j_{+}\in \mathbf{Y}_{\mathbf{x}}}\rrbracket  \pm [\widehat{j_{-}\in \mathbf{Y}_{\mathbf{x}}}\rrbracket\big]\leq 2\cdot Var[\widehat{j_{+}\in \mathbf{Y}_{\mathbf{x}}}]+2\cdot Var[\widehat{j_{-}\in \mathbf{Y}_{\mathbf{x}}}].$$}\normalsize Consequently, both $Var[\overline{\mathbf{x}}]$ and $Var[\underline{\mathbf{x}}]$ are not greater than $2\cdot\mathbb{E}\big[\sum_{j\in[d],\ b\in\{-1,1\}}|\widehat{\llbracket j_b\in \mathbf{Y}_{\mathbf{x}}\rrbracket }-\llbracket j_b\in \mathbf{Y}_{\mathbf{x}}\rrbracket |^2\big] = O(\frac{ds}{n\epsilon^2}).$

In many scenarios (e.g., federated gradient averaging), statisticians pay more attention to the mean value $\overline{\mathbf{x}}$. In this section, we analyze the pitfalls of the Collision mechanism for mean estimation and propose the correlated Collision mechanism (termed as CoCo), which obeys the negative correlation between $\llbracket j_+ \in \mathbf{Y}_{\mathbf{x}}\rrbracket $ and $\llbracket j_- \in \mathbf{Y}_{\mathbf{x}}\rrbracket $ so as to reduce estimation error.

\subsection{True/False/Opposite Collision Rate}
Recall that in the Collision mechanism, when $j_b \in \mathbf{Y}_{\mathbf{x}}$ or $j_{b} \notin \mathbf{Y}_{\mathbf{x}}$ holds, we have $\mathrm{P}[ H(j_b)=z] = \frac{e^\epsilon}{\Omega}$ and $\mathrm{P}[H(j_b)=z] = \frac{1}{t}$ respectively. We denote such conditional collision probabilities over the outputting domain as true/false/opposite collision rate (for $j\in [d]$ and $b\in \{+,-\}$):
\begin{equation*}
\begin{aligned}
&& P_{t}& := \mathbb{P}[H(j_b)=z\ \ |\ \ j_b \in \mathbf{Y}_{\mathbf{x}}], &\\
&& P_{f}& := \mathbb{P}[ H(j_b)=z\ \ |\ \ j_b \notin \mathbf{Y}_{\mathbf{x}}\ and\ j_{-b} \notin \mathbf{Y}_{\mathbf{x}}], &\\
&& P_{o}& := \mathbb{P}[ H(j_b)=z\ \ |\ \  j_{-b} \in \mathbf{Y}_{\mathbf{x}}].&\\
\end{aligned}
\end{equation*}

The variance of the mean estimator can be expressed as $Var\big[\widehat{\llbracket j_+\in \mathbf{Y}_{\mathbf{x}}\rrbracket }-\widehat{\llbracket j_-\in \mathbf{Y}_{\mathbf{x}}\rrbracket }\big]=\frac{Var[\llbracket H(j_+)=z\rrbracket -\llbracket H(j_-)=z\rrbracket ]}{(P_t-P_o)^2}$, which mainly depends on the discrepancy between the true/opposite collision rate. Meanwhile, in the Collision mechanism, we have $P_o\equiv P_f$ and $\frac{P_t}{P_o} < e^{\epsilon}$.

\subsection{Mechanism Design}
To maximize the discrepancy between the true/opposite collision rate and thus reduce the variance of the mean estimator, the CoCo mechanism aims to achieve $\frac{P_t}{P_o}\geq \frac{P_t}{P_f}$.
To accomplish this goal, we enforce stronger negative correlation between $\llbracket H(j_+)=z\rrbracket $ and $\llbracket H(j_-)=z\rrbracket $. 

Assuming the size of the outputting domain $t$ is even, we use two hash functions: $H_1: [d] \mapsto [\frac{t}{2}]$ and $H_2: \mathcal{Y} \mapsto \{-1,+1\}$. For any $j_b \in \mathbf{Y}_{\mathbf{x}}$, the overall hash function $H: \mathcal{Y} \mapsto [t]$ on $j_b$ is defined as:
$$H(j_b) := H_1(j)+\frac{b\cdot H_2(j_+)+1}{2}\cdot\frac{t}{2}.$$
Then, we assign the entry $H(j_b)$ in the output domain with a high relative probability $e^{\epsilon}$ and the entry $2\cdot H_1(j)+\frac{t}{2}-H(j_b)$ with a low relative probability $1$. The overall procedure of CoCo for a single user is summarized in Algorithm \ref{alg:coco}. Here, the sub-procedure $RandomPermute$ uniformly randomizes the order of elements in the given list or set, while the sub-procedure $Sum$ calculates the summation of weights in the provided list.

When $s> 1$, the $H_1(j)$ may conflict with each other for non-zero entries $\{j\ |\ \mathbf{x}_{j} \neq 0\}$. For every $k, k+\frac{t}{2}$ bucket pair ($k\in [\frac{t}{2}]$), we simply overwrite relative probabilities when there are conflicts (at line $8$-$11$ in Algorithm \ref{alg:coco}). To ensure that true/false/opposite collision rate is the same for every $j\in [d]$, the order of non-zero entries in $\mathbf{x}$ is firstly randomly permuted (line $5$  in Algorithm \ref{alg:coco}). To ensure that the normalization factor $\Omega=s\cdot (\epsilon+1)+(t-2\cdot s)$ is consistent for all possible inputs and hash functions, as in the Collision mechanism, the extra probability related to conflicted entries is uniformly redistributed to the remaining unassigned bucket pairs (at line $14$-$20$). The final output $z$ is then sampled according to relative probabilities of each outputting entry.

\begin{algorithm}[t]
\renewcommand\baselinestretch{1.0}\selectfont
    \caption{CoCo Randomizer}
    \label{alg:coco}
    \begin{algorithmic}[1]
        \Require A numerical data $\mathbf{x}\in \{-1,0,1\}^d$ with $s$ non-zero entries, privacy budget $\epsilon$, outputting domain size $t \in \mathbb{Z}^+$ that $t \geq 2s+2$ and $t\ mod\ 2= 0$, hash functions $\mathcal{H}_1$ and $\mathcal{H}_2$.
        \Ensure A private view $z \in [t]$ that satisfies $\epsilon$-LDP.
        \State{\color{gray} $\rhd$ Initialization}
        \State{select hash function $H_1: [d]\mapsto [\frac{t}{2}]$ from  $\mathcal{H}_1$} \edit{uniformly at random}
        \State{select hash function $H_2: \mathcal{Y} \mapsto \{-1,1\}$ from  $\mathcal{H}_2$} \edit{uniformly at random}
        \State{$W = \{0\}^t$}
        \State{$\mathbf{Y}'_\mathbf{x}=RandomPermute(\mathbf{Y}_\mathbf{x})$}
        \State{\color{gray} $\rhd$ Assign relative weights}
		\For{$j_b \in \mathbf{Y}'_\mathbf{x}$}
			\State{$H(j_b)= H_1(j)+\frac{b\cdot H_2(j_+)+1}{2}\cdot\frac{t}{2}$}
			\State{$W_{H(j_b)}=e^\epsilon$}
			\State{$H'(j_b)=2\cdot H_1(j)+\frac{t}{2}-H(j_b)$}
			\State{$W_{H'(j_b)}=1$}
		\EndFor
		\State{$\Omega=(e^\epsilon+1)\cdot s+t-2\cdot s$}
		\State{$w = \frac{\Omega-Sum(W)}{t-2\cdot Sum(W)/(e^\epsilon+1)}$}
		\For{$k \in [t/2]$}
		    \If{$W_k = 0$ and $W_{k+t/2} = 0$}
		        \State{$W_k = w$}

                \State{$W_{k+t/2} = w$}
		    \EndIf
		\EndFor
        \State{\color{gray} $\rhd$ Sampling with relative weights}
        \State{sampling one element $z \in [t]$ with probability $\mathbb{P}[z=k]=\frac{W_k}{\Omega}$}
        \\\Return{$(H_1, H_2, z)$}
    \end{algorithmic}
\end{algorithm}

\begin{algorithm}
    \renewcommand\baselinestretch{1.0}\selectfont
    \caption{CoCo Estimator}
    \label{alg:cocoestimator}
    \begin{algorithmic}[1]
        \Require A private view \edit{$(H_1, H_2, z)$} of unknown numerical data $\mathbf{x}_i$.
        \Ensure Estimators of $\llbracket j_+ \in \mathbf{Y}_{\mathbf{x}_i}\rrbracket +\llbracket j_- \in \mathbf{Y}_{\mathbf{x}_i}\rrbracket $ and $\llbracket j_+ \in \mathbf{Y}_{\mathbf{x}_i}\rrbracket -\llbracket j_- \in \mathbf{Y}_{\mathbf{x}_i}\rrbracket $.
        \For{$j \in [d]$}
            \For{$b \in \{-1,+1\}$}
                \State{$H(j_b)= H_1(j)+\frac{b\cdot H_2(j_+)+1}{2}\cdot\frac{t}{2}$}
            \EndFor
            \State{\color{gray} $\rhd$ Estimator of $\llbracket j_+ \in \mathbf{Y}_\mathbf{x}\rrbracket +\llbracket j_- \in \mathbf{Y}_\mathbf{x}\rrbracket $}
            \State{$\widehat{\underline{\mathbf{x}}}_{i,j}=\frac{\llbracket H((j_+)=z\rrbracket +\llbracket H((j_-)=z\rrbracket -2\cdot P_f}{P_t+P_o-2\cdot P_f}$}
            \State{\color{gray} $\rhd$ Estimator of $\llbracket j_+ \in \mathbf{Y}_\mathbf{x}\rrbracket -\llbracket j_- \in \mathbf{Y}_\mathbf{x}\rrbracket $}
            \State{$\widehat{\overline{\mathbf{x}}}_{i,j}=\frac{\llbracket H(j_+)=z\rrbracket -\llbracket H(j_-)=z\rrbracket }{P_t-P_o}$}
        \EndFor
        \\\Return{$\{\widehat{\underline{\mathbf{x}}}_{i,j}, \widehat{\overline{\mathbf{x}}}_{i,j}\}_{j\in [d]}$}
    \end{algorithmic}
\end{algorithm}

For better illustration, we depict an example of applying the $(d=10,s=3,\epsilon=\log 2,t=8)$-CoCo mechanism on numerical data $\mathbf{x}=[0,0,1,0,-1,0,0,0,-1,0]$ in Figure \ref{fig:coco}. It shows a case when overwrite/conflict happens for hash functions $(H_1,H_2)$.

\begin{figure}
\vspace*{-1.0em}
\begin{center}
\centerline{\includegraphics[width=73mm]{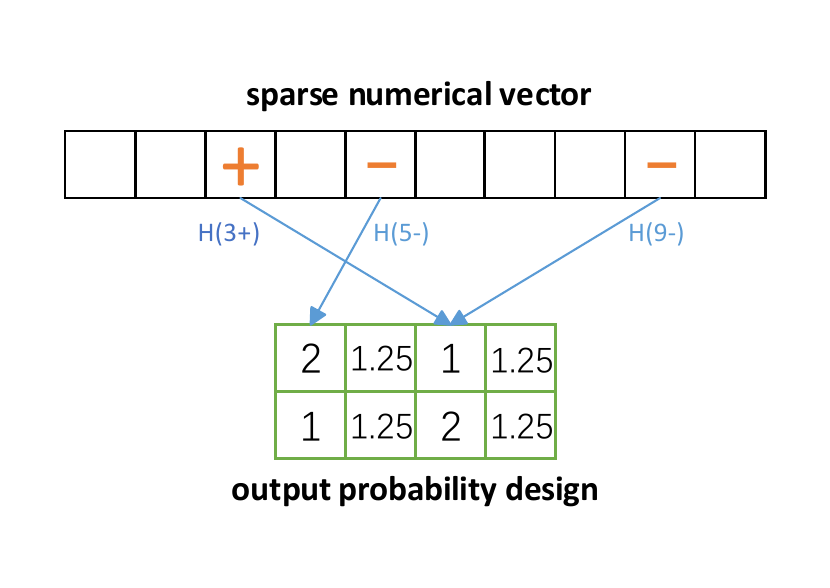}}
\vspace*{-2.0em}
\caption{An illustration of the CoCo mechanism with hash conflicts/overwrite, where $d=10$, $s=3$, $t=8$ and $\epsilon=\log(2)$.}
\label{fig:coco}
\end{center}
\vspace*{-0.0em}
\end{figure}

We will now proceed to examine the behavior of CoCo in terms of true, false, and opposite collision rates. Let $P_{ow}$ denote the probability that a non-zero entry $j$ is overwritten on the outputting domain by other entries. The true collision rate is then:
$$P_{t} = P_{ow}\cdot \frac{e^\epsilon+1}{2\cdot\Omega}+(1-P_{ow})\cdot \frac{e^\epsilon}{\Omega},$$
the false collision rate is:
$$P_{f} = \frac{1}{t},$$
and the opposite collision rate is:
$$P_{o} = P_{ow}\cdot \frac{e^\epsilon+1}{2\cdot\Omega}+(1-P_{ow})\cdot \frac{1}{\Omega}.$$

The formula of the $P_{ow}$ is close-formed. Separately considering the permuted order of a non-zero entry $j$, since there are exactly $s$ entries in $\mathbf{Y}_\mathbf{x}$, the probability that the entry $j$ ranks $k$ among $s$ entries is $\frac{1}{s}$ (for $k\in [s]$). When $j$ is the $k$-th entry, there are remaining $s-k$ entries that have not been hashed, thus the conflict/overwrite probability is $1-(\frac{t-2}{t})^{s-k}$. Therefore, we have:
\begin{equation}
P_{ow}=1-\frac{1}{s}\sum_{k=1}^s(\frac{t-2}{t})^{s-k}=1-\frac{t^{s}-(t-2)^{s}}{2 t^{s-1}\cdot s}.
\end{equation}
When $t\geq 2s+2$ and $\epsilon>0$, $P_{o}$ is always less than $P_{f}$, and thus provides opportunity for more accurate mean estimation. As a comparison, the original Collision mechanism has $P_f\equiv P_o$.

\textbf{Mean Estimator.\ \ }
We now proceed to derive an unbiased estimator of the $j$-th mean value $\frac{1}{n}\sum_{i\in [n]} \mathbf{x}_{i,j}$, which equals to $\frac{1}{n}\sum_{i\in [n]} \llbracket j_+\in \mathbf{Y}_\mathbf{x_i}\rrbracket -\llbracket j_-\in \mathbf{Y}_\mathbf{x_i}\rrbracket $. Observe that when some non-zero entry $j'_{b'}$ ($j'\neq j$) overwrites bucket pair $(H_1(j),H_1(j)+\frac{t}{2})$, since the hash function $H_2$ is uniform pseudo-randomly, we have $\mathbb{E}\big[\llbracket H(j_+)=z\rrbracket -\llbracket H(j_-)=z\rrbracket\big]=0$. Otherwise, when no overwrite happens to $(H_1(j),H_1(j)+\frac{t}{2})$, we have $\mathbb{E}\big[\llbracket H(j_+)=z\rrbracket -\llbracket H(j_-)=z\rrbracket \big]=\frac{\llbracket j_+\in \mathbf{Y}_\mathbf{x}\rrbracket -\llbracket j_-\in \mathbf{Y}_\mathbf{x}\rrbracket }{e^\epsilon /\Omega -1/\Omega}$. Combining two results together, we have $\mathbb{E}\big[\llbracket H(j_+)=z\rrbracket -\llbracket H(j_-)=z\rrbracket \big]=P_{ow}\cdot 0 +(1-P_{ow})\cdot \frac{\llbracket j_+\in \mathbf{Y}_\mathbf{x}\rrbracket -\llbracket j_-\in \mathbf{Y}_\mathbf{x}\rrbracket }{e^\epsilon /\Omega -1/\Omega}.$
Therefore, we arrived an unbiased estimator of the $j$-th mean value as (the $H^i$ is hash function used by user $i$):
$$\frac{1}{n}\sum_{i\in [n]}\frac{\llbracket H^i(j_+)=z^i\rrbracket -\llbracket H^i(j_-)=z^i\rrbracket }{P_{t}-P_{o}}.$$

\textbf{Non-missing Frequency Estimator.\ \ } In the key-value data aggregation, statisticians are also interested in the non-missing frequency of each key: $\underline{\mathbf{x}}_j= \frac{1}{n}\#\{\mathbf{x}_{i,j}\ |\ \mathbf{x}_{i,j}\ for\ i\in[n]\ and\ \mathbf{x}_{i,j}\neq 0\}.$
When $j$ is a non-missing entry in $\mathbf{x}$, since $H(j+)\neq H(j_-)$, it is obvious that $\mathbb{P}[z=H(j_+)\ or\ z=H(j_-)]=P_t+P_o=\frac{e^\epsilon+1}{2\cdot \Omega}$; When $j$ is a missing entry in $\mathbf{x}$, we have $\mathbb{P}[z=H(j_+)\ or\ z=H(j_-)] =2\cdot P_f$. Consequently, according to the transition matrix of the CoCo mechanism, we get:
$$\mathbb{E}\Big[\frac{\llbracket z=H(j_+)\rrbracket +\llbracket z=H(j_-)\rrbracket -2\cdot P_f}{\frac{e^\epsilon+1}{2\cdot \Omega}-2\cdot P_f}\Big]=\llbracket\mathbf{x}_{i,j}\neq 0\rrbracket .$$
An unbiased estimator of the $j$-th non-missing frequency $\underline{\mathbf{x}}_j$ is thus:
$$\frac{1}{n}\sum_{i\in [n]}\frac{\llbracket H^i(j_+)=z^i\rrbracket +\llbracket H^i(j_-)=z^i\rrbracket -2\cdot P_f}{P_{t}+P_{o}-2\cdot P_f}.$$
We summarize these estimators in Algorithm \ref{alg:cocoestimator}, which relies on the transition probability matrix in Table \ref{tab:transition} concerning various events on the outputs given conditions in the inputs.

We now analyze the complexities of the proposed CoCo mechanism. On the user side, the computational cost is $O(s)$, and the communication cost is $O(\log t)=O(\epsilon+\log s)$. On the server side, the naïve approach in Algorithm \ref{alg:cocoestimator} that derives estimators for each $(H_1,H_2,z)$ has a computational cost of $O(n\cdot d)$, and a memory cost of $O(\log t)$. Alternatively, one can first record frequencies of every $(H_1,H_2,z) \in \mathcal{H}_1\times \mathcal{H}_2 \times [t]$, and then summarize $\llbracket H(j_b)=z\rrbracket $ with the frequency weight. Assuming the domain size of $\mathcal{H}_1\times \mathcal{H}_2$ is constant, this approach has a computational cost of $n+t\cdot d=O(n+d s e^\epsilon)$ and a memory cost of $O(s e^\epsilon)$.

\begin{table}
\renewcommand{\arraystretch}{1.3}
\caption{Conditional probabilities about the input \& output for $j\in [d]$ and $b\in \{-1,+1\}$. The probability takes into account the randomness of selecting hash function, uniform pseudo-randomness of the hash functions, and the randomness of sampling $z$.}
\label{tab:transition}
\centering
\begin{tabular}{|c||c|c|}
\hline
 & $\llbracket H(j_b)=z\rrbracket $ & $\llbracket H(j_{-b})=z\rrbracket $ \\
\hline
\hline
$j_b \in \mathbf{Y}_\mathbf{x}$ & $P_t$ & $P_o$ \\
\hline
$j_{-b} \in \mathbf{Y}_\mathbf{x}$ & $P_o$  & $P_t$ \\
\hline
$j_b \notin \mathbf{Y}_\mathbf{x}\ and\ j_{-b} \notin \mathbf{Y}_\mathbf{x} $ & $P_f$ & $P_f$ \\
\hline
\end{tabular}
\end{table}

\subsection{Theoretical Analyses}
In this part, we provide privacy and accuracy guarantees of the CoCo mechanism. The $\epsilon$-LDP guarantee of the mechanism is given in Proposition \ref{pro:coldp}.

\begin{proposition}\label{pro:coldp}
The $(d,s,\epsilon,t)$-CoCo mechanism in Algorithm \ref{alg:coco} satisfies $\epsilon$-LDP for numerical vector data.
\end{proposition}
\begin{proof}
First, the normalization factor $\Omega$ in the CoCo mechanism is the same for any input $\mathbf{x}$ and any hash functions $H_1\in \mathcal{H}_1, H_2\in \mathcal{H}_2$. Second, due to the identicalness of selecting hash functions (i.e., follow the same distribution), we only need to consider the private view $(H_1,H_2,z)$ given fixed $H_1,H_2$. Third, given $H_1$ and $H_2$, the relative probabilities of every outputting entry range from $1.0$ to $e^\epsilon$. Since $t\geq 2s+2$ implies the $w$ at line $14$ is lower than $(e^\epsilon+1)/(2\Omega)$ but never lower than $1/\Omega$, then for any $a\in [t]$ and any inputs $\mathbf{x},\mathbf{x'}\in \mathcal{X}^s$, we have $\frac{\mathbb{P}[z=a | \mathbf{x},H_1,H_2]}{\mathbb{P}[z=a | \mathbf{x'},H_1,H_2]}\leq \frac{e^\epsilon/\Omega}{1.0/\Omega}\leq e^\epsilon$.
\end{proof}

\subsubsection{Mean Squared Error}
With the outputting domain size parameter $t$ fixed in the CoCo mechanism, its estimation errors of various estimators (see Algorithm \ref{alg:cocoestimator}) are presented in Lemma \ref{lemma:cocomse}.

\begin{lemma}\label{lemma:cocomse}
For the $(d,s,\epsilon,t)$-CoCo mechanism, the mean squared errors of estimators are:
\small
\begin{equation}\label{eq:nmfmse}
\begin{aligned}
\sum_{j=1}^d |\widehat{\underline{\mathbf{x}}}_{j}-{\underline{\mathbf{x}}}_{j}|^2 =\frac{s (P_t+P_o)(1-P_t-P_o)+(d-s)2 P_f (1-2 P_f)}{(P_t+P_o-2 P_f)^2},
\end{aligned}
\end{equation}
\begin{equation}\label{eq:meanmse}
\begin{aligned}
\sum_{j=1}^d |\widehat{\overline{\mathbf{x}}}_{j}-{\overline{\mathbf{x}}}_{j}|^2 =\frac{s((P_t+P_o)-(P_t-P_o)^2) +(d-s)(2P_f)}{(P_t-P_o)^2}.
\end{aligned}
\end{equation}
\normalsize
\end{lemma}

\begin{proof}
See Appendix \ref{app:cocomseformula}
\end{proof}

Based on the error formulation, we further choose parameter $t$ in Theorem \ref{the:cocomse} (see Appendix \ref{proof:cocomse} for proof). Consequently, the mean squared errors are approximately minimized and reach the optimal $O(\frac{d s }{\epsilon^2})$ bound.

\begin{theorem}[Mean Squared Error Bounds]\label{the:cocomse}
When $\epsilon=O(1)$, takes as an input \edit{$\mathbf{x}$}, the $(d,s,\epsilon,t)$-CoCo mechanism with $t=\lceil e^\epsilon s+5s \rceil$ satisfies
\begin{equation}\label{eq:nmfmsebound}
    \begin{aligned}
    && &\sum\nolimits_{j=1}^d |\widehat{\underline{\mathbf{x}}}_{j}-{\underline{\mathbf{x}}}_{j}|^2 \leq O\big(\frac{d s }{\epsilon^2}\big);&&
    \end{aligned}
\end{equation}
the $(d,m,\epsilon,t)$-CoCo mechanism with $t=\lceil e^\epsilon s+s+2\rceil$ satisfies 
\begin{equation}\label{eq:meanmsebound}
    \begin{aligned}
    && &\sum\nolimits_{j=1}^d |\widehat{\overline{\mathbf{x}}}_{j}-{\overline{\mathbf{x}}}_{j}|^2 \leq O\big(\frac{d s}{\epsilon^2}\big).&&
    \end{aligned}
\end{equation}
\end{theorem}

\begin{figure}
\begin{center}
\centerline{\includegraphics[width=89mm]{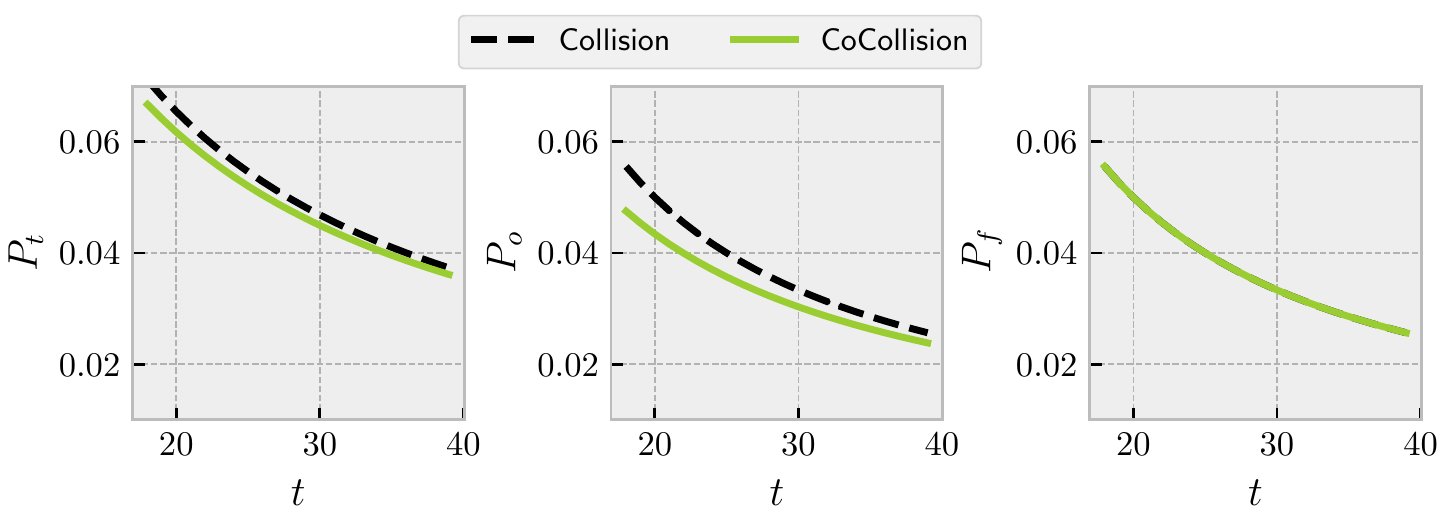}}
\vskip -0.1in
\caption{The $P_t$,$P_n$ and $P_f$ that varies with outputting domain size $t$ when $d=128$ and $s=8$. Compared to the original Collision mechanism, the opposite collision rate $P_o$ in CoCo is significantly lower and is smaller than the $P_f$.}
\label{fig:varyrates}
\end{center}
\vspace*{-0.5em}
\end{figure}

To illustrate the impact of parameter $t$, we plot the variation of $P_t$/$P_o$/$P_f$ and mean squared errors in Figure \ref{fig:varyrates} and \ref{fig:varyvars}, in comparison to the previously proposed Collision mechanism. According to the variance bounds of the sum and difference of two variables (see the beginning of Section \ref{sec:meanmechanism}), the Collision mechanism also satisfies the same error bound. By designing the CoCo mechanism to have $P_o < P_f$, the constant factor in its error is reduced.

\begin{figure}
\begin{center}
\centerline{\includegraphics[width=89mm]{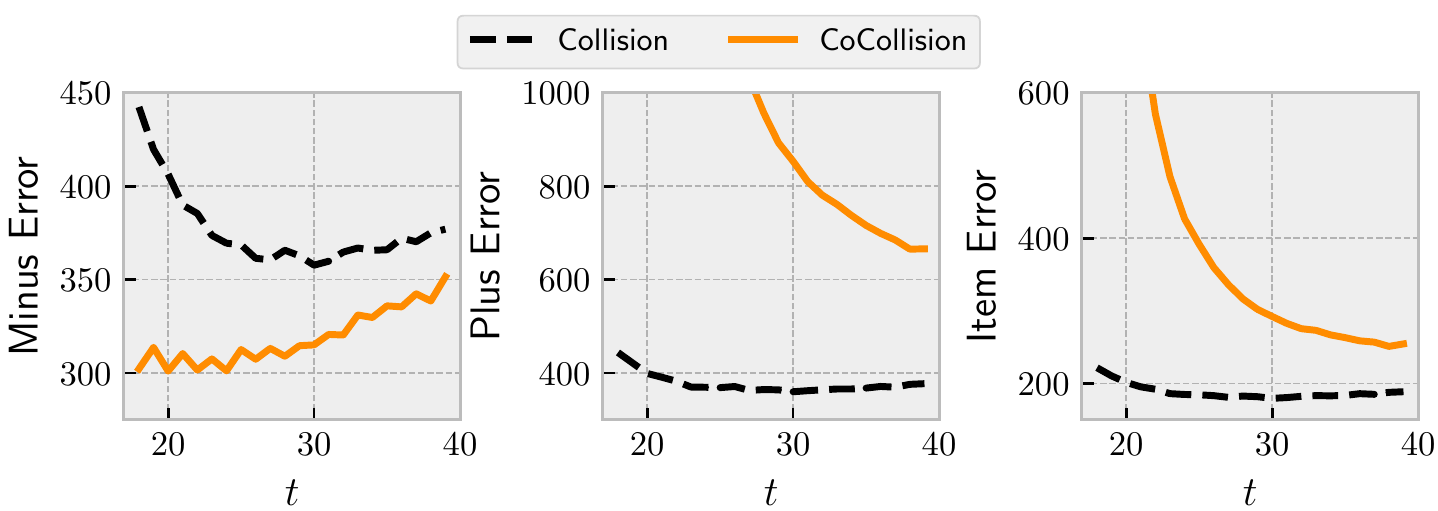}}
\vskip -0.1in
\caption{The estimation errors that varies with outputting domain size $t$ with $d=128$, $s=8$, $n=1$, and $\epsilon=0.5$. The \textit{Minus Error} denotes $\sum_{j=1}^d |\widehat{\overline{\mathbf{x}}}_{j}-{\overline{\mathbf{x}}}_{j}|^2$; the \textit{Plus Error} denote $\sum_{j=1}^d |\widehat{\underline{\mathbf{x}}}_{j}-{\underline{\mathbf{x}}}_{j}|^2$; the \textit{Item Error} denote $\sum_{j_b \in \mathcal{Y}}| \llbracket\widehat{j_b\in \mathbf{Y}_\mathbf{x}}\rrbracket -\llbracket j_b\in \mathbf{Y}_\mathbf{x}\rrbracket |^2$. All results are the average values of $10,000$ independent experiments. The CoCo has about $20\%$ lower MSE errors on the mean estimator.} 
\label{fig:varyvars}
\end{center}
\vspace*{-0.5em}
\end{figure}

\subsubsection{Maximum Absolute Error}\label{subsec:maebound}
In this section, we derive the expected maximum absolute error of the proposed CoCo mechanism for mean estimation and demonstrate that it is rate-optimal.
Based on the (discrete) probability distributions of the observed variable $\llbracket H(j_+)=z\rrbracket -\llbracket H(j_-)=z\rrbracket$, we present the maximum absolute error bounds of the CoCo mechanism in Theorem \ref{the:meanbound} (see Appendix \ref{proof:cocomae} for proof). This implies that the error is bounded by ${O}(\frac{1}{\epsilon}\sqrt{\frac{s \log d}{n}})$.

\begin{theorem}[Maximum Absolute Error of Mean Estimation in CoCo]\label{the:meanbound}
With privacy budget $\epsilon=O(1)$, for mean value estimation on $n$ users, the error due to Algorithm \ref{alg:coco} and \ref{alg:cocoestimator} is bounded by
$$\max\nolimits_{j=1}^{d}{|\widehat{\overline{\mathbf{x}}}_j-\overline{\mathbf{x}}_j|} \leq O\Big(\sqrt{\frac{s\log(d/\beta)}{\epsilon^2 n}}\Big)$$
with probability $1-\beta$ over the randomness of the user-specific hash functions and the randomization in Algorithm \ref{alg:coco}.
\end{theorem}

Recently, for mean estimation of $s$-sparse numerical vectors, \cite{zhou2022locally} analyzed lower bounds on the maximum absolute error $\max_{j=1}^d |\widehat{\overline{\mathbf{x}}}_{j}-{\overline{\mathbf{x}}}_{j}|^2$ under $\epsilon$-LDP. We restate the minimax lower bound $O(\frac{1}{\epsilon}\sqrt{\frac{s \log{d/s}}{n}})$ in Theorem \ref{the:maeminimax}, which follows definitions in Section \ref{sec:bound}. Combining the upper error bounds in Theorem \ref{the:meanbound}, we can conclude that the CoCo mechanism is minimax optimal (when $s\leq \sqrt{d}$) under the measurement of maximum absolute error. 

\begin{theorem}[Lower Bounds of Mean Estimation \cite{zhou2022locally}]\label{the:maeminimax}
For the numerical vector mean estimation problem, for any $\epsilon$-LDP mechanism, there exists a universal constant $c> 0$ such that for all $\epsilon\in (0,1]$,
$$\mathfrak{M}_n(\underline{\theta}(\mathcal{P}), \|\cdot\|_{\infty}, \epsilon)\geq \min \Big\{c\cdot \frac{1}{\epsilon}\sqrt{\frac{s \log{d/s}}{n}}, 1\Big\}.$$ 
\end{theorem}

\subsection{Privacy Amplification in the Shuffle Model}
In this section, we consider privacy amplification of the CoCo mechanism in the shuffle model. Since CoCo has a similar probability design as the Collision, let $\mathcal{R}$ denote the $(d,s,\epsilon,t)$-CoCo mechanism (assuming $t>s$), and let $\alpha=\frac{s(e^\epsilon-1)}{s e^\epsilon+t-s}$, then for any neighboring datasets $D,D'$, we also have:
\begin{equation*}
\mathcal{D}(\mathcal{S}\circ\mathcal{R}(D)\| \mathcal{S}\circ\mathcal{R}(D'))\leq \mathcal{D}(P_{\alpha}\|Q_{\alpha}).
\end{equation*}
The equality holds when there are no hash collisions for all user data in $D$ and $D'$ (requires $t\geq 4s$).

%% file: experiments.tex
\begin{figure}
\begin{center}
\centerline{\includegraphics[width=82mm]{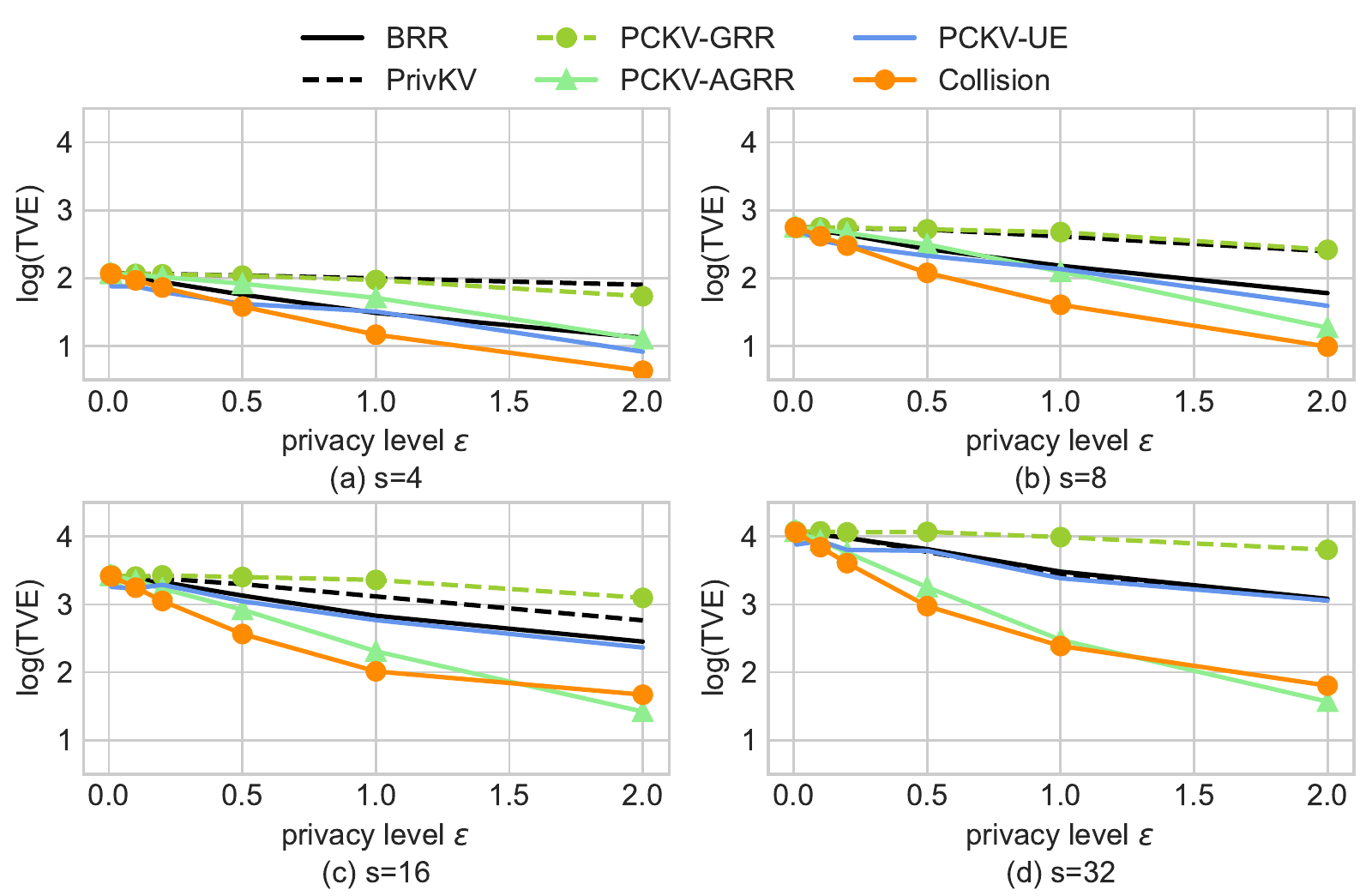}}
\vskip -0.1in
\caption{Frequency estimation TVE results on $n=100,000$ users with dimension $d=256$ when sparsity $s$ ranges from $4$ to $32$.}
\label{fig:tve256n100000}
\end{center}
\vspace*{-1.5em}
\end{figure}

\section{Experiments}\label{sec:experiments}
In this section, we mainly evaluate the statistical efficiency of the proposed Collision/CoCo mechanism for $\epsilon$-LDP numerical vector aggregation. Competing mechanisms include the PCKV mechanism with unary encoding as the base randomizer \cite{gu2019pckv} (denoted as PCKV-UE), the PrivKV mechanism \cite{Ye2019PrivKVKD}, the PCKV mechanism with generalized randomized response as the base randomizer (denoted as PCKV-GRR), its privacy amplified version (denoted as PCKV-AGRR), and the succinct mean estimation protocol \cite{zhou2022locally} (denoted as SUCCINCT). Since the performances of all these mechanisms are data-independent, it is sufficient to utilize synthetic datasets for fair evaluation. The parameters of synthetic datasets are listed as follows (default values are in \textbf{bold} form), covering most cases encountered in real-world applications:
\begin{enumerate}
\item[i.] Number of users $n$: 10,000 and \textbf{100,000}.
\item[ii.] Dimension $d$: 256 and \textbf{512}.
\item[iii.] Sparsity parameter $s$: 4, 8, 16, and 32.
\item[iv.] Privacy budget $\epsilon$: 0.001, 0.01, 0.1, 0.2, 0.4, 0.8, 1.0, 1.5, and 2.0.
\end{enumerate}
Since competing mechanisms are data-independent (i.e., estimation errors are irrelevant of true values), during each simulation, the numerical vector of each user is independently and randomly generated, the non-zero entries are uniformly and randomly selected from $d$ dimensions, and each dimension has an equal probability of being $-1$ or $1$.

\begin{figure}[t]
\vskip 0.0in
\begin{center}
\centerline{\includegraphics[width=82mm]{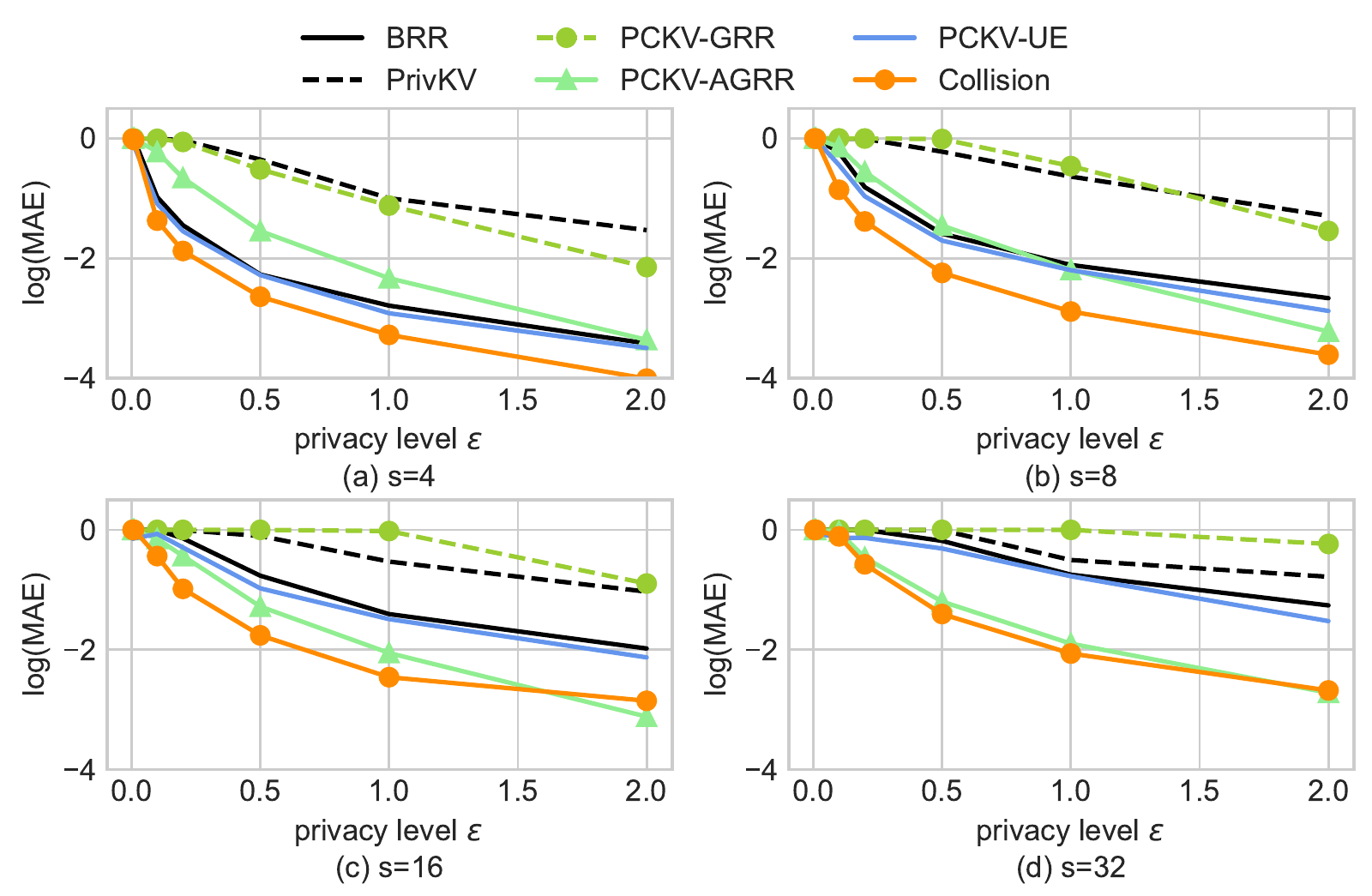}}
\vskip -0.1in
\caption{Frequency estimation MAE results on $n=100,000$ users with dimension $d=256$ when sparsity $s$ ranges from $4$ to $32$.}
\label{fig:mae256n100000}
\end{center}
\vspace*{-1.5em}
\end{figure}

\subsection{Evaluation Metric}
As frequency estimators are basic statistics for both the non-missing frequency estimation and mean estimation, we evaluate mechanisms with metrics TVE and MAE on $\llbracket j_b\in \mathbf{Y}_{\mathbf{X}}\rrbracket$. The total variation error (TVE) of frequency estimation is defined as:
\small
$$\text{TVE}=\sum_{j\in[d],\ b\in\{-1,1\}}\big|\widehat{\llbracket j_b\in \mathbf{Y}_{\mathbf{X}}\rrbracket}-\llbracket j_b\in \mathbf{Y}_{\mathbf{X}}\rrbracket\big|,$$
\normalsize
and the maximum absolute error (MAE) is defined as:
\small
$$\text{MAE}=\max_{j\in[d],\ b\in\{-1,1\}}\big|\widehat{\llbracket j_b\in \mathbf{Y}_{\mathbf{X}}\rrbracket}-\llbracket j_b\in \mathbf{Y}_{\mathbf{X}}\rrbracket\big|.$$
\normalsize
For mean estimation, we use TVE and MAE metrics in the similar way.

Since the $\frac{1}{s}$-scaled frequencies lie in the $2d$-dimensional probability simplex, the estimated frequencies are post-processed by projecting them into the $\Delta_{2d}$-simplex \cite{wang2015projection}. All experimental results are the mean natural logarithm value of 100 repeated simulations.

\subsection{Frequency Estimation}
In this section, we measure the performance of frequency estimation under various settings, such as varying sparsity, dimension, and number of users.
\subsubsection{Effects of sparsity $s$}
Assuming that there are $n=100,000$ users and the dimension is $d=256$. When the number of non-zero entries in numerical vectors varies from $4$ to $32$, the TVE/MAE error results are presented in Figure \ref{fig:tve256n100000} and Figure \ref{fig:mae256n100000}, respectively. The PCKV-UE mechanism improves upon the PrivKV in extremely sparse cases, but for other cases (e.g., $s=32$), the PCKV-UE and PrivKV mechanisms have similar performances. The Collision mechanism outperforms all competing mechanisms in almost all cases significantly, and on average reduces more than $60\%$ errors. When the sparsity parameter and privacy budget are large (e.g., $s\geq 16$, and $\epsilon=2$), the performance gap between PCKV-AGRR and Collision decreases.

\subsubsection{Effects of dimension $d$}
Assuming that there are $n=100,000$ users, but the dimension now increases to $d=512$. When the number of non-zero entries in numerical vectors still varies from $4$ to $32$, the TVE and MAE results are shown in Figure \ref{fig:tve512n100000} and Figure \ref{fig:mae512n100000}, respectively. Compared to cases with $d=256$ (i.e., TVE results in Figure \ref{fig:tve256n100000} and MAE results in Figure \ref{fig:mae256n100000}), it is evident that the TVE/MAE value grows with approximately $\sqrt{d}$.

\begin{figure}[t]
\begin{center}
\centerline{\includegraphics[width=82mm]{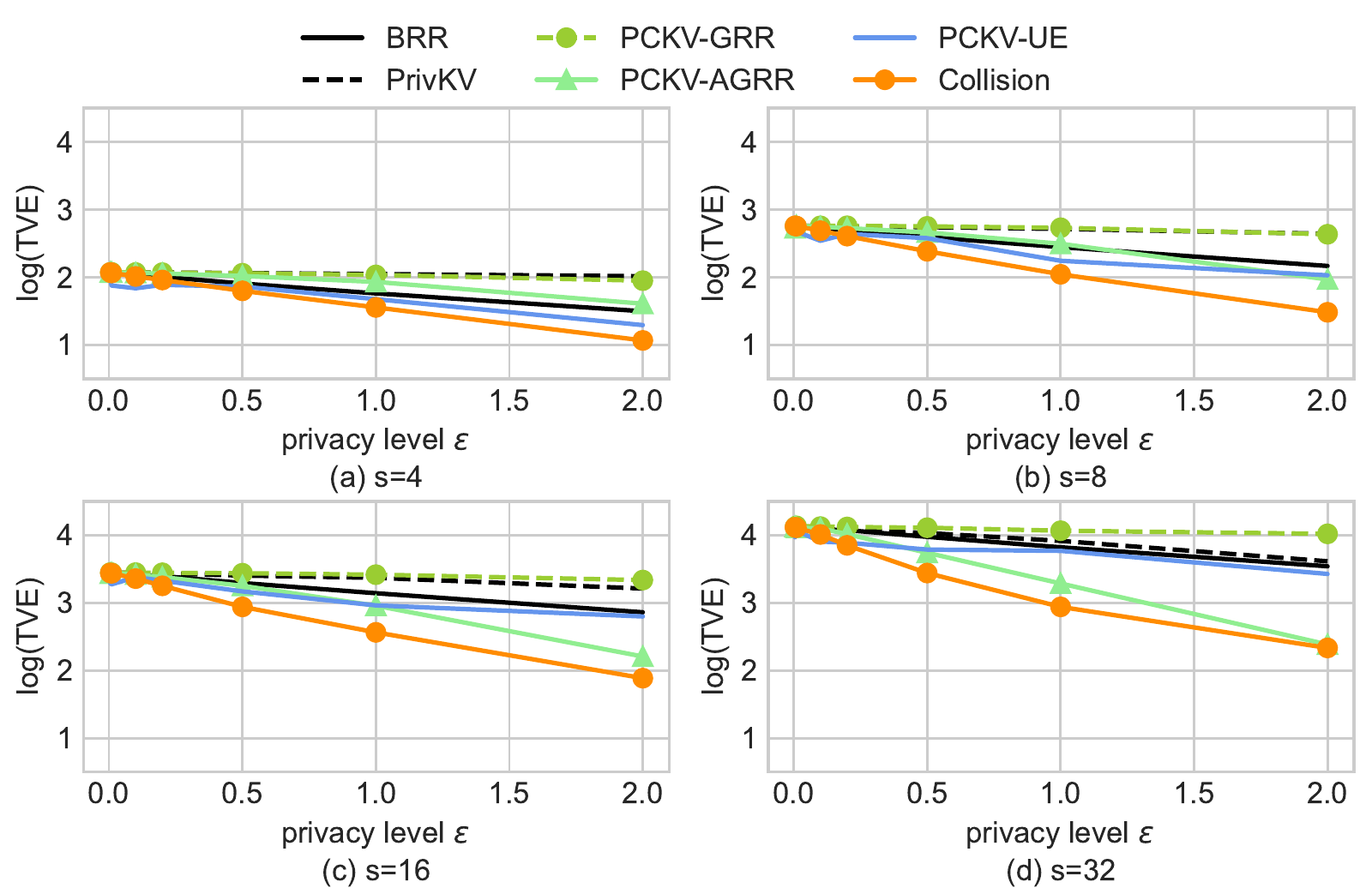}}
\vskip -0.1in
\caption{Performance evaluation of frequency estimation in terms of TVE on $n=100,000$ users with dimension $d=512$, when sparsity $s$ ranges from $4$ to $32$.}
\label{fig:tve512n100000}
\end{center}
\vspace*{-1.5em}
\end{figure}

\begin{figure}[t]
\vskip 0.0in
\begin{center}
\centerline{\includegraphics[width=82mm]{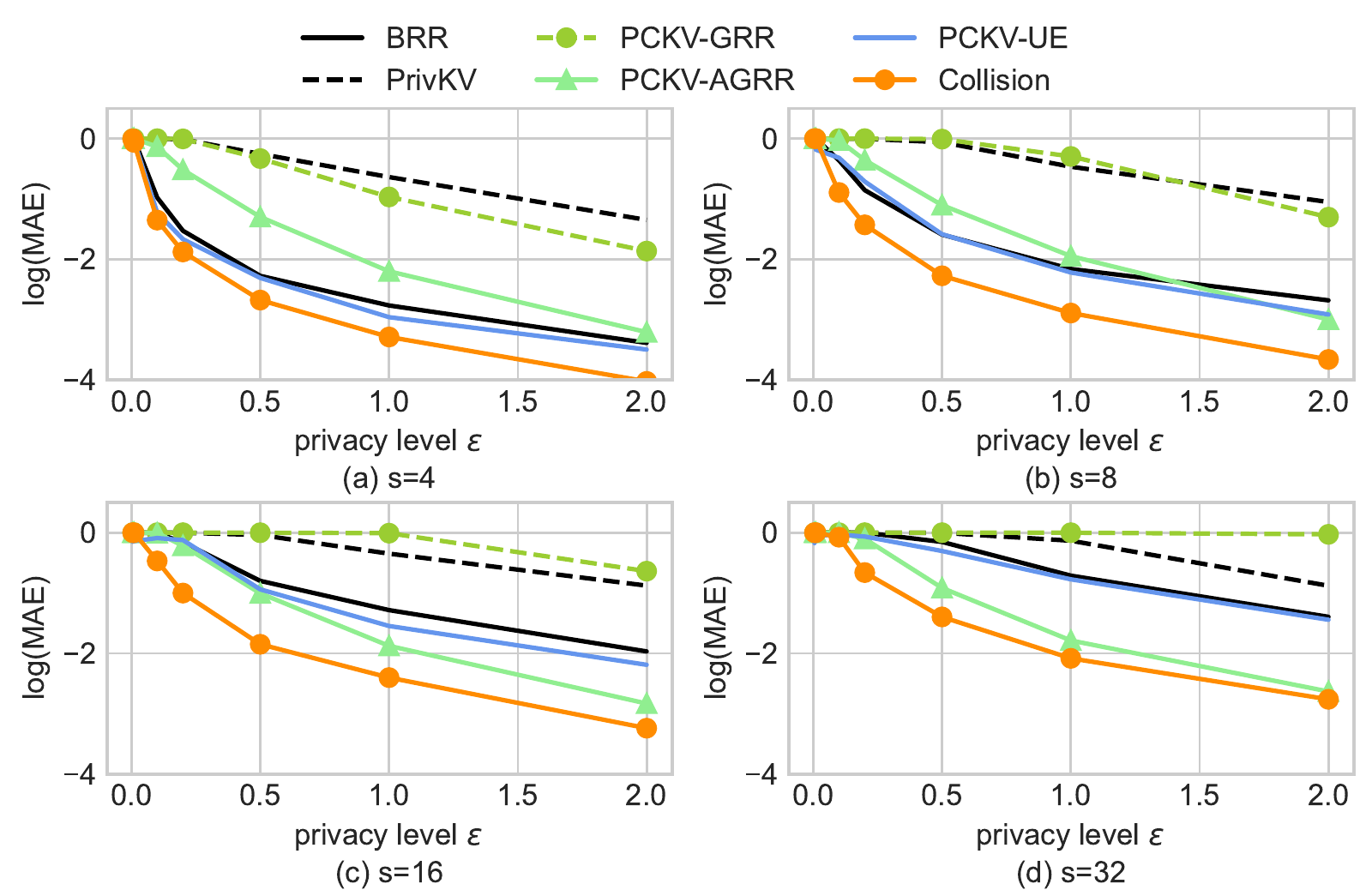}}
\vskip -0.1in
\caption{Performance evaluation of frequency estimation in terms of MAE on $n=100,000$ users with dimension $d=512$, when sparsity $s$ ranges from $4$ to $32$.}
\label{fig:mae512n100000}
\end{center}
\vspace*{-1.5em}
\end{figure}

\subsubsection{Effects of dimension $d$}
Assuming that there are $n=100,000$ users, but the dimension now increases to $d=512$. When the number of non-zero entries in numerical vectors still varies from $4$ to $32$, the TVE and MAE results are shown in Figure \ref{fig:tve512n100000} and Figure \ref{fig:mae512n100000}, respectively. Compared to cases with $d=256$ (i.e., TVE results in Figure \ref{fig:tve256n100000} and MAE results in Figure \ref{fig:mae256n100000}), it is evident that the TVE/MAE value grows with approximately $\sqrt{d}$.

\subsubsection{Effects of Number of Users $n$}
Assuming that there are only $n=10,000$ users and the dimension is $d=256$. When the number of non-zero entries in numerical vectors varies from $4$ to $32$, the TVE and MAE results are listed in Figure \ref{fig:tve256n10000} and Figure \ref{fig:mae256n10000}, respectively. Compared to the case with $n=100,000$ (i.e., Figure \ref{fig:tve256n100000} and Figure \ref{fig:mae256n100000}), the TVE/MAE value is about $\sqrt{100000/10000}$ times larger (i.e., decreases with approximately $\sqrt{n}$).

\begin{figure}[t]
\begin{center}
\centerline{\includegraphics[width=82mm]{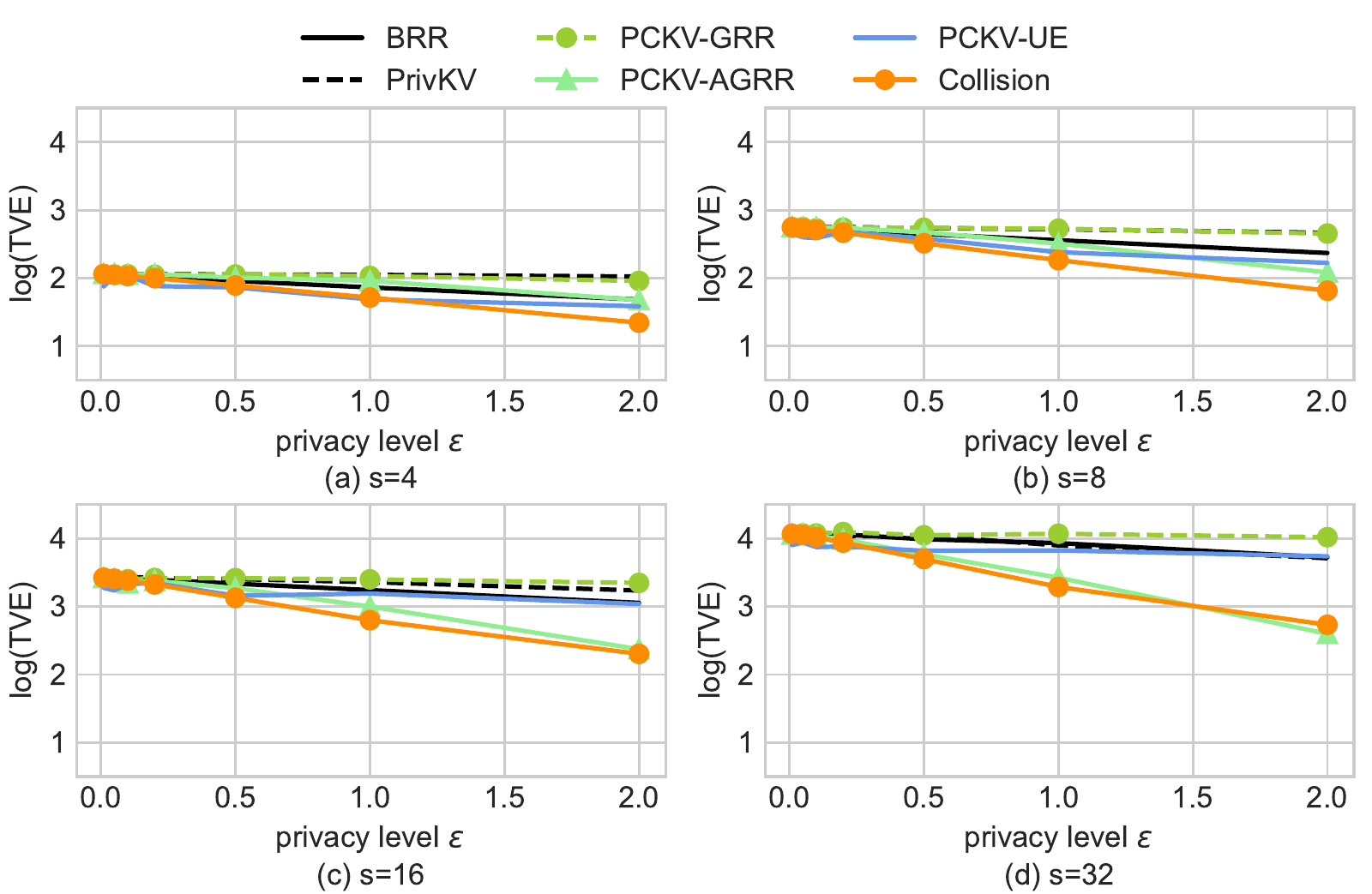}}
\vskip -0.1in
\caption{Frequency estimation TVE results on $n=10,000$ users with dimension $d=256$ when sparsity $s$ ranges from $4$ to $32$.}
\label{fig:tve256n10000}
\end{center}
\vspace*{-1.5em}
\end{figure}

\begin{figure}[t]
\begin{center}
\centerline{\includegraphics[width=82mm]{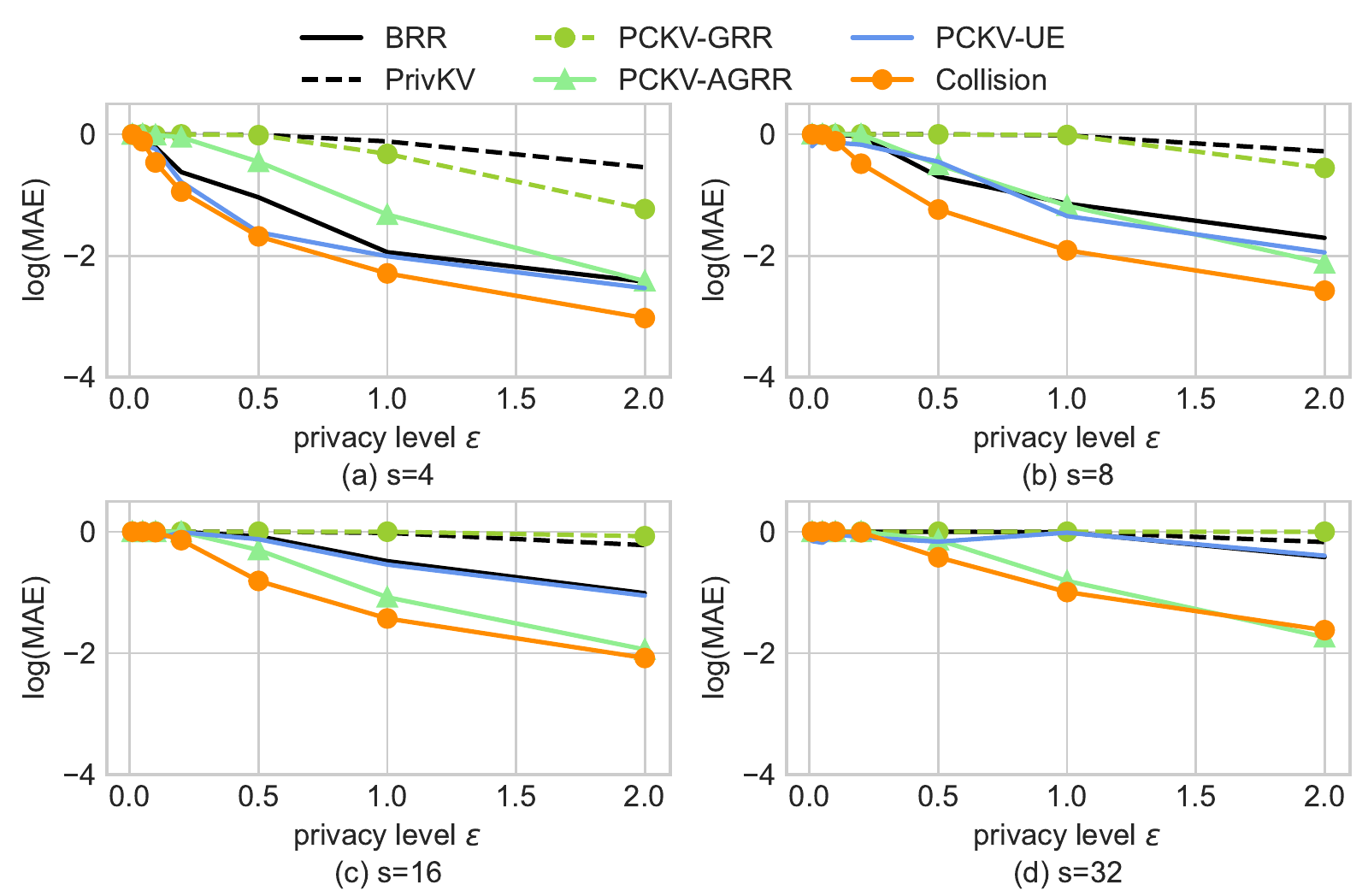}}
\vskip -0.1in
\caption{Frequency estimation MAE results on $n=10,000$ users with dimension $d=256$ when sparsity $s$ ranges from $4$ to $32$.}
\label{fig:mae256n10000}
\end{center}
\vspace*{-1.5em}
\end{figure}
\subsubsection{After shuffling}
In the shuffle model, given a global privacy goal $(\epsilon_c, \delta)$, the local privacy budget approximately scales with $\tilde{O}(\epsilon_c\sqrt{n/\log(1/\delta)})$. It is observed that the Collision mechanism outperforms existing approaches across all privacy regions. By combining the theoretical results that provide precisely tight privacy accounting for the Collision (see Theorem \ref{the:upper} and Figure \ref{fig:shuffle}), its performance in the shuffle model is assured.

\subsection{Mean Estimation}\label{subsec:meanexperiments}
This section presents an evaluation of the performances of competing mechanisms for mean estimation.

\subsection{Effects of post-processing}
Except the SUCCINCT \cite{zhou2022locally} that is designed only for mean estimation, other mechanisms' frequency estimators can be post-processed to the $\Delta_{2d}$-simplex, and then be utilized for deriving the corresponding mean estimators. For fair comparison, we present experimental results both without and with post-processing in Figure \ref{fig:mae256n100000meanraw} and \ref{fig:mae256n100000mean} respectively. It is observed that the Collision mechanism outperforms existing approaches by about $30\%$ in almost all settings, and the CoCo mechanism further reduce more than $15\%$ error compared to the Collision. This confirms our theoretical analyses on the CoCo, which forces the opposite collision rate $P_o$ to be smaller than the false collision rate $P_f$ by design. Even without post-processing, the SUCCINCT mechanism is less competitive when privacy budget $\epsilon$ or sparsity $s$ is relatively large, as its estimation error due to clip bias grows.

\begin{figure}[t]
\begin{center}
\centerline{\includegraphics[width=82mm]{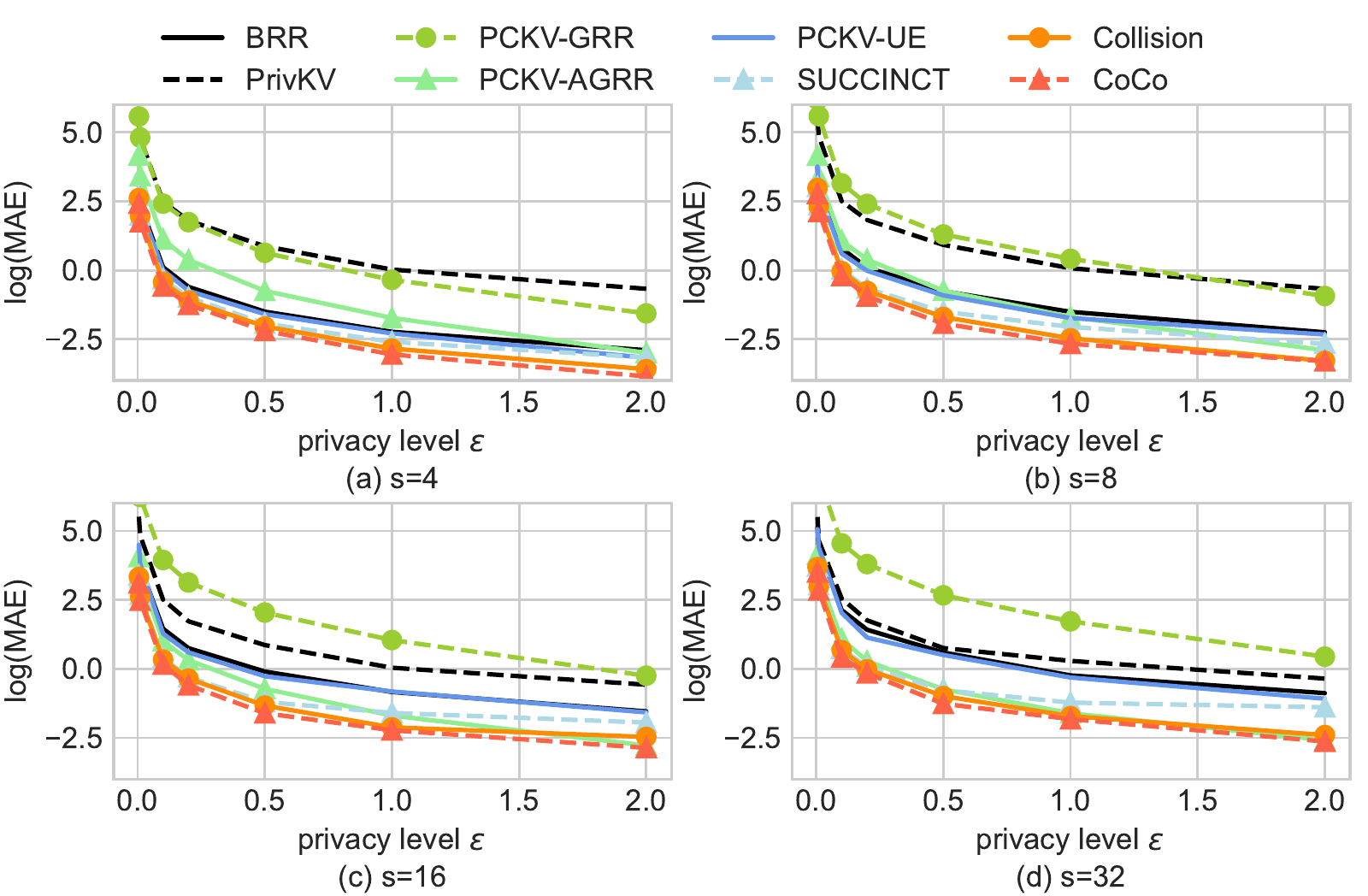}}
\vskip -0.1in
\caption{Mean estimation MAE results without post-processing on $n=100,000$ users with dimension $d=256$ when sparsity $s$ ranges from $4$ to $32$.}
\label{fig:mae256n100000meanraw}
\end{center}
\vspace*{-1.5em}
\end{figure}

\begin{figure}[t]
\begin{center}
\centerline{\includegraphics[width=82mm]{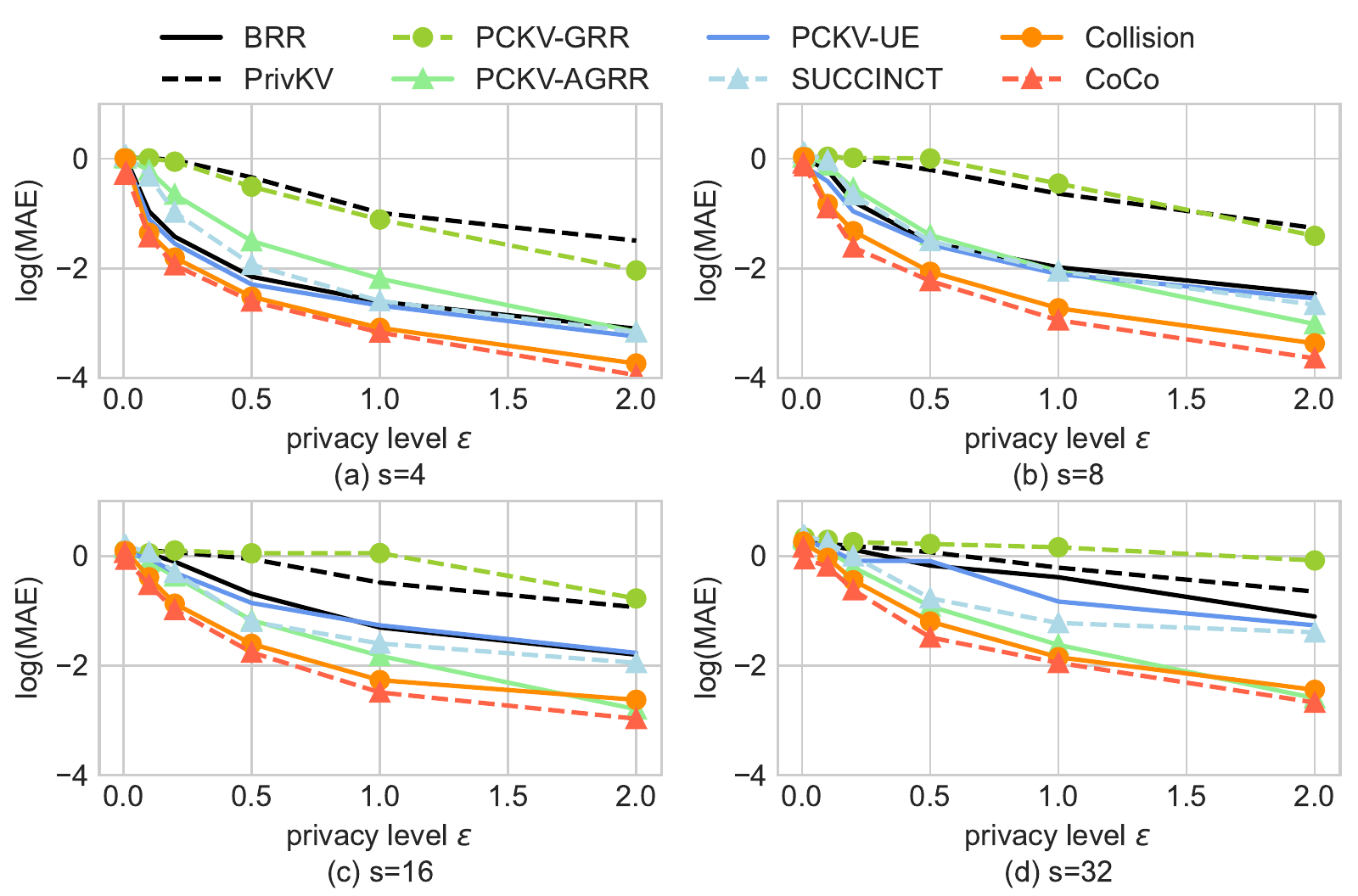}}
\vskip -0.1in
\caption{Mean estimation MAE results with post-processing on $n=100,000$ users with dimension $d=256$ when sparsity $s$ ranges from $4$ to $32$.}
\label{fig:mae256n100000mean}
\end{center}
\vspace*{-1.5em}
\end{figure}

\subsubsection{Effects of sparsity $s$}
Simulated with $n=100,000$ users and dimension $d=512$, the number of non-zero entries in numerical vectors varies from $4$ to $32$. The TVE results are listed in Figure \ref{fig:tve512n100000mean}. The Collision/CoCo mechanisms outperform existing approaches by about $30\%$ in almost all settings. When the numerical vector gets denser (i.e., $\frac{d}{s}$ gets smaller) and the privacy budget is large, the performance gap between PCKV-AGRR and Collision/CoCo decreases.

\begin{figure}[t]
\begin{center}
\centerline{\includegraphics[width=82mm]{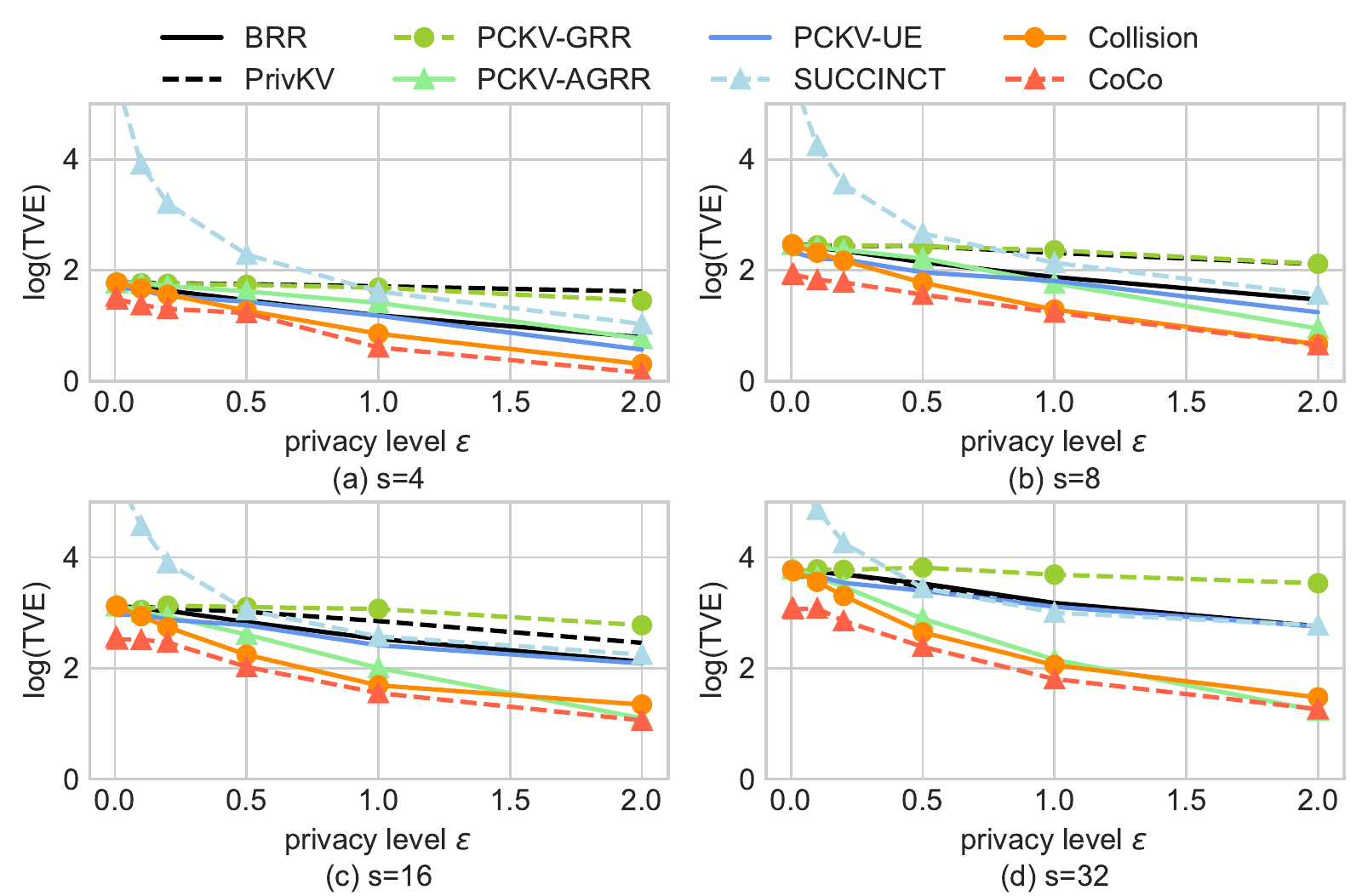}}
\vskip -0.1in
\caption{Mean estimation MAE results with post-processing on $n=100,000$ users with dimension $d=256$ when sparsity $s$ ranges from $4$ to $32$.}
\label{fig:tve256n100000mean}
\end{center}
\vspace*{-1.5em}
\end{figure}

\subsubsection{Effects of dimension $d$}
Simulated with $n=100,000$ users and dimension $d=512$, the results of TVE is shown in Figure \ref{fig:tve512n100000mean}. Compared to cases of $d=256$ (i.e., TVE results in Figure \ref{fig:tve256n100000mean}), it is easy to observe that the TVE grows roughly with $\sqrt{d}$.

\begin{figure}[t]
\begin{center}
\centerline{\includegraphics[width=82mm]{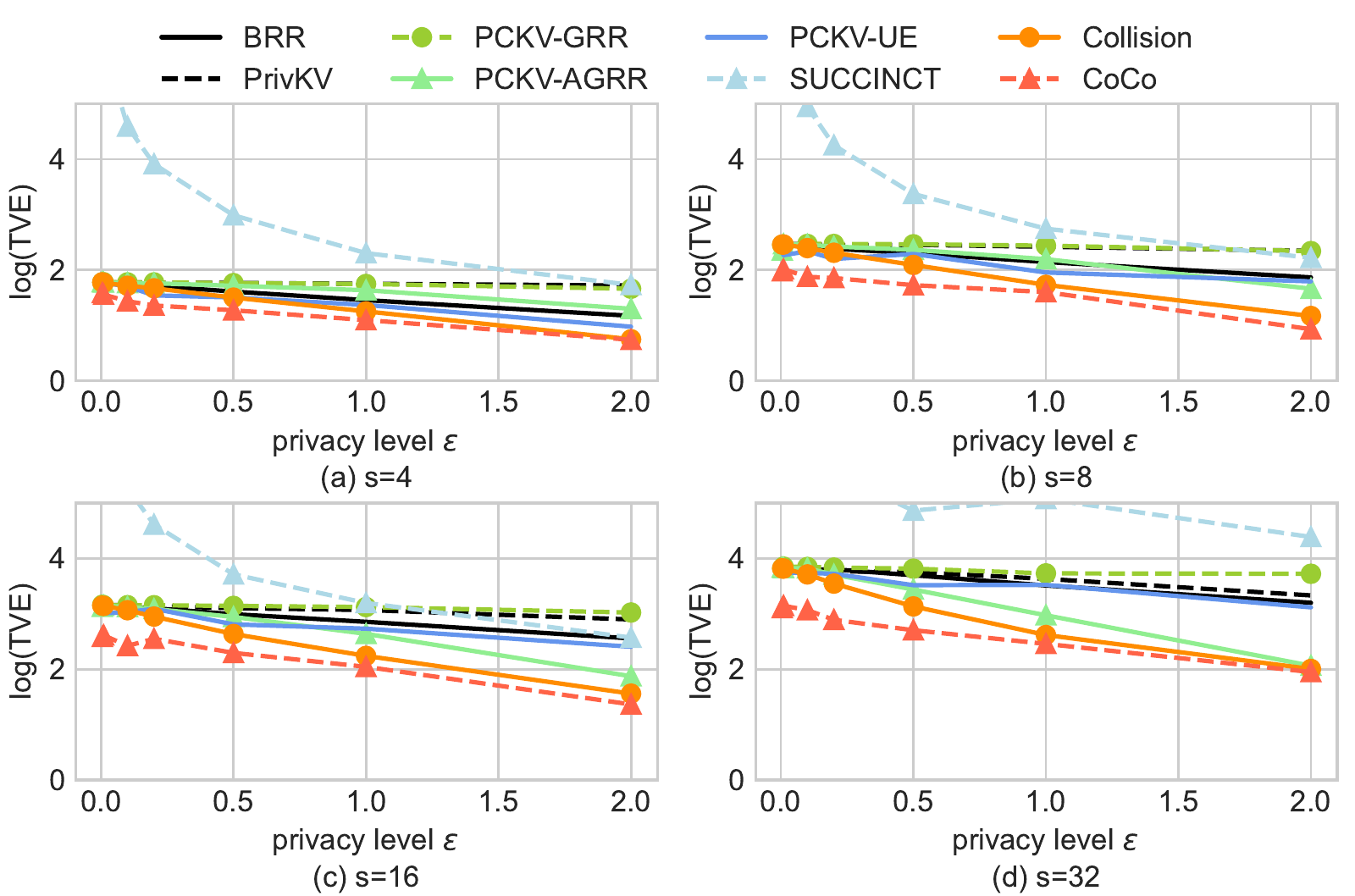}}
\vskip -0.1in
\caption{Mean estimation TVE results with post-processing on $n=100,000$ users with dimension $d=512$ when sparsity $s$ ranges from $4$ to $32$.}
\label{fig:tve512n100000mean}
\end{center}
\vspace*{-1.5em}
\end{figure}

\subsubsection{Effects of number of users $n$}
Simulated with $n=10,000$ users and dimension $d=256$, the TVE results are listed in Figure \ref{fig:tve256n10000mean}. Compared to the case of $n=100,000$ (i.e. Figure \ref{fig:tve256n100000mean}, the TVE is about $\sqrt{100000/10000}$ times larger.

\begin{figure}[t]
\begin{center}
\centerline{\includegraphics[width=82mm]{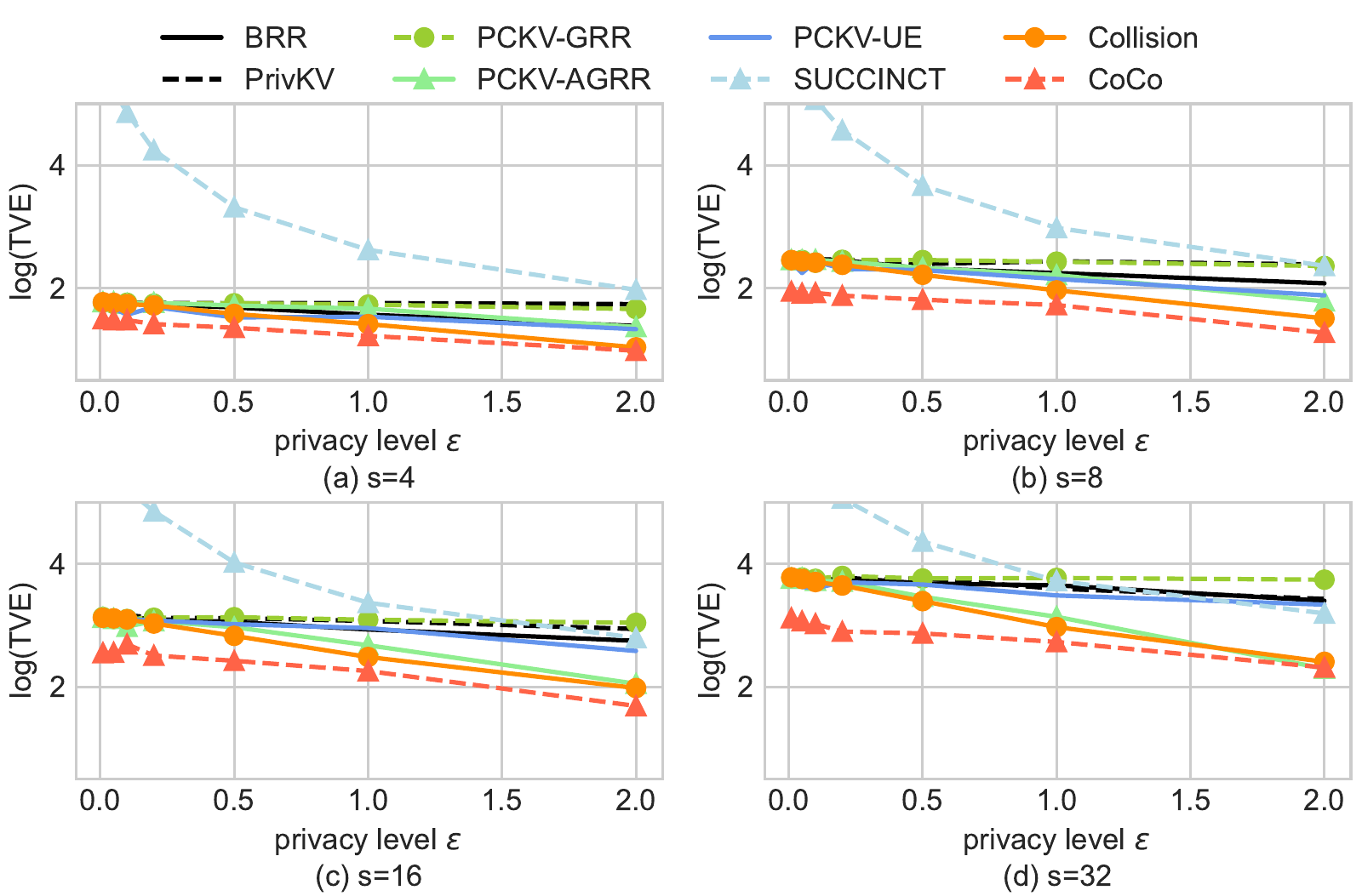}}
\vskip -0.1in
\caption{Mean estimation TVE results with post-processing on $n=10,000$ users with dimension $d=256$ when sparsity $s$ ranges from $4$ to $32$.}
\label{fig:tve256n10000mean}
\end{center}
\vspace*{-1.5em}
\end{figure}

\subsection{Experimental summary}
Through experimental evaluation, we conclude that the Collision mechanism outperforms existing approaches in all cases for frequency estimation (especially when $1\ll s\ll d$), and the CoCo mechanism further improve accuracy by about $15\%$ for mean estimation. Their performance gaps confirm our theoretical analyses on error bounds.

%% file: conclusion.tex
\vspace{0.6em}
\section{Conclusion}\label{sec:conclusion}
Within the local and shuffle model of differential privacy, this work has presented several \emph{simple yet optimal} results for the problem of numerical vector statistical estimation, which has its applications in federated learning and key-value data aggregation.
We provided tight minimax error bounds for locally private estimation on numerical vectors. Our proof relies on a novel decomposition technique for data domain with sparse structure and an application of the local private version of Assouad methods. Given that existing approaches are suffering gaps form the minimax error bound, we further design an optimal mechanism based on frequency estimation, and then give an efficient implementation with $O(s)$ or $O(\log s)$ computation/communication complexity. Specifically for mean estimation, we propose the CoCo mechanism, which utilizes the negative correlation in frequencies to further reduce estimation error. To break the error bound of LDP, we consider numerical vector estimation in the shuffle model, and derive tight privacy amplification bounds for proposed mechanisms.
Experimental results show $30\%$-$60\%$ error reduction of our proposed optimal mechanisms when compared with current approaches. 

\textbf{Future researches.} While this work studied numerical vector analyses in the single-message shuffle model, it is promising to further improve utility with multi-message protocols, at the cost of more communication overheads (e.g., tens of messages) per user.

%% file: appendix.tex
\nobalance
\appendices

\section{Proof of Minimax Lower Bounds}\label{app:minimax}
According to Lemma \ref{lemma:ldpassouad}, in order to derive a good lower bound, we need to construct a well hamming-separated class of distributions, and simultaneously minimize the maximum possible marginal distance. Our proof follows four steps:
\begin{enumerate} 
    \item[1)]{\edit{Constrain} numerical vector data to decomposable cases. Following analyzing procedures will utilize this decomposability for simplifying analysis on numerical vector into multiple categorical cases. Specifically, we assume that the dimensionality $d$ is divisible by the sparsity parameter $s$, that is, we can define an integer value $l:=\frac{d}{s}$. We can then separate $d$ dimensions to $s$ buckets, and the $a$-th bucket is $B_a=\{\mathbf{x}_{a\cdot l+1},...,\mathbf{x}_{a\cdot l+l}\}$ ($0\leq a\leq s-1$). As a special form of numerical vector data, we consider cases when each bucket $B_a$ has exactly $1$ non-zero entry and call such numerical vector data a decomposable one. Such decomposability, along with independence among buckets, allows us to simplify the proof.}
    \item[2)]{Construct $2\delta$-Hamming separation distributions. Follow standard procedure of Assouad method, we set \edit{$\mathcal{V}_a\in\{-1,1\}^{l}$} for each bucket $B_a$ ($a \in [0,s-1]$}) and define a hypercube as  \edit{$\mathcal{V}=\Pi_{a=0}^{s-1}\mathcal{V}_a$}. Fixing $\delta\in [0,1]$, for $\nu \in \mathcal{V}$, separately consider each bucket, we define \edit{$\theta_{\nu_a}\in \mathbb{R}^{2l}$} be the multinomial distribution for bucket $B_a$ as:
    $$\theta_{\nu_a}:= \frac{s}{2d}\mathbf{1}+\delta\frac{s}{2d}\begin{bmatrix}\mathcal{V}_a\\ -\mathcal{V}_a\end{bmatrix},$$
    where the $j$-th element of $\theta_{\nu_a}$ represents the probability that the $a\cdot l +j$-th element of the numerical vector is $1$ for $j \in [1:l]$, while the $(l+j)$-th element of $\theta_{\nu_a}$ represents the probability that the $a\cdot l +j$-th element of the numerical vector is $-1$ for $j \in [1:l]$.

    Assuming independence among buckets, we then define the probability distribution over the universe $\mathcal{X}^s$ as a product distribution \edit{$\Pi_{a=0}^{s-1}\mathcal{\theta}_{\nu_a}$}. The item distribution is hence  \edit{$\theta_\nu=[\theta_{\nu_0}\ \theta_{\nu_2}\ ...\ \theta_{\nu_{s-1}}]$.} For any estimator \edit{$\widehat{\theta}=[\widehat{\theta}_{\nu_0}\ \widehat{\theta}_{\nu_2}\ ...\ \widehat{\theta}_{\nu_{s-1}}]$}, by defining $\widehat{\nu}_a=sign(\widehat{\theta}_{\nu_a})$ for $a \in [0,s-1]$, we have lower bound on separation:
    \begin{equation}\label{eq:separation}
    \|\widehat{\theta}-\theta_\nu\|_2^2\geq \frac{\delta^2 s^2}{d^2}\sum_{j=1}^{l}\sum_{a=1}^s\llbracket\widehat{\nu}_{a_j}\neq \nu_{a_j}\rrbracket.
    \end{equation}
    
    \item[3)]{Bound the maximum distance of induced marginal distributions. We now turn to bounding sums of integrals $\int_{\mathcal{X}^s}\gamma(x)(d P_{+j}(x)-d P_{-j}(x))$, and prove following inequality:
    $$\sup_{\gamma\in \mathbb{B}_\infty(\mathcal{X}^s) }\sum_{j=1}^{d}(\int_{\mathcal{X}^s}\gamma(x)(d P_{+j}(x)-d P_{-j}(x)))^2\leq \frac{8\delta^2 s}{d}.$$
    Actually, by construction, $P_{+j}$ is a joint distribution $\Pi_{a=1}^s \big[\frac{s}{2d}\mathbf{1}+\frac{s\delta}{2d}[e_{j\mod{l}}^\top -e_{j\mod{l}}^\top]^\top\llbracket\lfloor j/l \rfloor = a\rrbracket\big] \in \Delta_{2l}$ and similarly for $P_{-j}$, where $e_j\in\{0,1\}^{l}$ denote the $j$-th standard basis vector. Due to the interleaving structure of the $s$-dimensional distribution $P_{+j}$ and $P_{+j}$, for any $\gamma \in [-1,1]^{2d}$, we have:
    $$\sum_{a=1}^s\sum_{j=1}^{l}(\int_{\mathcal{X}^s}\gamma(x)(d P_{+j}(x)-d P_{-j}(x)))^2 \leq \frac{8\delta^2 s}{d},$$
    that is, assigning $\gamma$ along one of the dimension maximizes the overall integral.}
    \item[4)]{Bound the minimax risks. Applying Lemma \ref{lemma:ldpassouad} and \edit{substituting} the hamming separation parameter $\delta$ as $\frac{\delta^2 s^2}{d^2}$, we have:
    $$\max_{\nu \in \mathcal{V}} \mathbb{E}_{P_\nu}[\|\widehat{\theta}-\theta_\nu\|_2^2]\geq \frac{\delta^2 s^2}{d}[1-(4n(e^\epsilon-1)^2\delta^2s/d^2)^\frac{1}{2}].$$
    By choosing the parameter $\delta^2$ at $\min\{1, d^2/(16n(e^\epsilon-1)^2s)\}$, we have the lower bound of:
    $$\mathfrak{M}_n(\theta(\mathcal{P}), \|\cdot\|_2^2, \epsilon)\geq \min\{\frac{s^2}{4d},\frac{d s}{64n(e^\epsilon-1)^2}\}.$$}
\end{enumerate}

\section{Proof of Mixture Property of Collision}\label{app:mixture}
Let $R_1$, $R_1'$, and $R_i$ denote the probability distributions of $\mathcal{R}(x_1)$, $\mathcal{R}(x_1')$, and $\mathcal{R}(x_i)$, respectively, where $i\in [2:n]$. In this paper, we establish the existence of mixture distributions by means of construction. Specifically, we define the probability distributions of $\mathcal{Q}_1$, $\mathcal{Q}_1'$, $\mathcal{Q}_1^*$, $\mathcal{Q}_2$, $\ldots$, and $\mathcal{Q}_n$ in the following manner:
\begin{equation*}
    \mathcal{Q}_1[H,z] :=\left\{
    \begin{array}{@{}lr@{}}
        \frac{R_1[H,z]-R_1'[H,z]}{(e^{\epsilon}-1)\beta}, & \text{if } R_1[H,z]> R_1'[H,z];\\
        0, & \text{else,}
    \end{array}
    \right. 
\end{equation*}
\begin{equation*}
    \mathcal{Q}_1'[H,z] :=\left\{
    \begin{array}{@{}lr@{}}
        \frac{R_1'[H,z]-R_1[H,z]}{(e^{\epsilon}-1)\beta}, & \text{if } R_1[H,z]> R_1'[H,z];\\
        0, & \text{else,}
    \end{array}
    \right. 
\end{equation*}
\begin{equation*}
    \mathcal{Q}_1^*[H,z] :=\frac{\min\{R_1'[H,z],R_1[H,z]\}}{1-\beta-e^{\epsilon}\beta}-\frac{ |R_1[H,z]-R_1'[H,z]|}{(e^{\epsilon}-1)(1-\beta-e^{\epsilon}\beta)};
\end{equation*}
\begin{equation*}
    \mathcal{Q}_i[H,z] :=
        \frac{R_i[H,z]-\beta(\mathcal{Q}_1[H,z]+\mathcal{Q}_1'[H,z])}{1-2\beta}.
\end{equation*}

To establish Equations (\ref{eq:mix1}), (\ref{eq:mix2}), and (\ref{eq:mix3}), we first demonstrate that $\mathcal{Q}_1$ and $\mathcal{Q}_1'$ are valid probability distributions. For $\mathcal{Q}_1$, we observe that the probability $\mathcal{Q}_1[H,z]$ is non-negative everywhere, and $\sum_{(H,z)\in \mathcal{H}\times\mathcal{Z}} \mathcal{Q}_1[H,z] = \frac{\sum_{(H,z)\in \mathcal{H}\times\mathcal{Z}} |R_1[H,z]- R_1'[H,z]|}{2(e^{\epsilon}-1)\beta}=1$. Similarly, we establish that $\mathcal{Q}_1'$ is a valid probability distribution.

Next, we establish that $\mathcal{Q}_1^*$ is a valid distribution and demonstrate that Equations \ref{eq:mix1} and \ref{eq:mix2} hold. Specifically, since $R_1'[H,z]\leq e^{\epsilon} R_1[H,z]$ and $R_1[H,z]\leq e^{\epsilon} R_1'[H,z]$, it follows that $\mathcal{Q}_1^*[H,z]$ is non-negative. Furthermore, we have $R_1[H,z]=e^\epsilon \beta \mathcal{Q}_1[H,z] + \beta \mathcal{Q}_1'[H,z]+(1-\beta-e^\epsilon \beta)\mathcal{Q}_1[H,z]$ for all $H,z$, which implies that $\mathcal{R}(x_1^0)=e^{\epsilon} \beta \mathcal{Q}_1 + \beta \mathcal{Q}_1'+(1-\beta-e^{\epsilon} \beta)\mathcal{Q}_1^*$ (and Equation \ref{eq:mix1} holds).

Since $\mathcal{Q}_1$ and $\mathcal{Q}_1'$ are valid distributions and $e^{\epsilon} \beta+ \beta+(1-\beta-e^{\epsilon} \beta)=1$, it follows that $\mathcal{Q}_1^*$ is a valid distribution. Similarly, we demonstrate that Equation \ref{eq:mix2} holds.

We demonstrate the validity of the distribution $\mathcal{Q}_i$ for $i\in[2:n]$ and the Equation $\eqref{eq:mix3}$. Since either $\mathcal{Q}_1[H,z]$ or $\mathcal{Q}_1'[H,z]$ is equal to zero, we have $\mathcal{Q}_1[H,z]+\mathcal{Q}_1'[H,z]=\max(\mathcal{Q}_1[H,z],\mathcal{Q}_1'[H,z])$. Utilizing the property of $\epsilon$-LDP of $\mathcal{R}$, we obtain $R_i[H,z]\geq \max(R_1[H,z],R_1'[H,z])/e^\epsilon\ge(e^\epsilon \beta\mathcal{Q}_1[H,z]+e^\epsilon \beta\mathcal{Q}_1'[H,z])/e^\epsilon\geq \beta\mathcal{Q}_1[H,z]+ \beta\mathcal{Q}_1'[H,z]$. Given that $\beta+\beta+(1-2\beta)=1$, we conclude that $\mathcal{Q}_i$ is a valid distribution. The Equation $\eqref{eq:mix3}$ follows directly from the definition of $\mathcal{Q}_i$.

\section{Proof of Mean Squared Error Formulas of CoCo}\label{app:cocomseformula}
For Equation (\ref{eq:nmfmse}), we consider two cases separately: when $\llbracket j_+ \in \mathbf{Y}_\mathbf{x}\rrbracket =1 \text{ or } \llbracket j- \in \mathbf{Y}_\mathbf{x}\rrbracket =1$, and when $\llbracket j+ \in \mathbf{Y}_\mathbf{x}\rrbracket =0\text{ and } \llbracket j- \in \mathbf{Y}_\mathbf{x}\rrbracket =0$. In the first case, the variable $\llbracket H(j+)=z\rrbracket +\llbracket H(j_-)=z\rrbracket $ in Algorithm \ref{alg:cocoestimator} is a Bernoulli variable of success rate $P_t+P_o$, which implies that $Var[\llbracket H(j_+)=z\rrbracket +\llbracket H(j_-)=z\rrbracket ]=(P_t+P_o)(1-P_t-P_o)$ and $Var[\widehat{\overline{\mathbf{x}}}{j}]=\frac{(P_t+P_o)(1-P_t-P_o)}{(P_t+P_o-2P_f)^2}$. In the second case, the variable $\llbracket H(j+)=z\rrbracket +\llbracket H(j_-)=z\rrbracket $ is a Bernoulli variable of success rate $2P_f$, which implies that $Var[\llbracket H(j_+)=z\rrbracket +\llbracket H(j_-)=z\rrbracket ]=(2P_f)(1-2P_f)$ and $Var[\widehat{\overline{\mathbf{x}}}_{j}]=\frac{(2P_f)(1-2P_f)}{(P_t+P_o-2P_f)^2}$

Considering Equation (\ref{eq:nmfmse}), we analyze two cases separately: (i) $\llbracket j_+ \in \mathbf{Y}_\mathbf{x}\rrbracket =1$ or $\llbracket j_- \in \mathbf{Y}_\mathbf{x}\rrbracket =1$, and (ii) $\llbracket j_+ \in \mathbf{Y}_\mathbf{x}\rrbracket =0$ and $\llbracket j_- \in \mathbf{Y}_\mathbf{x}\rrbracket =0$. In the first case, the random variable $\llbracket H(j_+)=z\rrbracket +\llbracket H(j_-)=z\rrbracket$ in Algorithm \ref{alg:cocoestimator} is a Bernoulli variable with a success rate of $P_t+P_o$. Therefore, the variance of $\llbracket H(j_+)=z\rrbracket +\llbracket H(j_-)=z\rrbracket$ is $Var[\llbracket H(j_+)=z\rrbracket +\llbracket H(j_-)=z\rrbracket ]=(P_t+P_o)(1-P_t-P_o)$, and the variance of $\widehat{\underline{\mathbf{x}}}{j}$ is $Var[\widehat{\underline{\mathbf{x}}}{j}]=\frac{(P_t+P_o)(1-P_t-P_o)}{(P_t+P_o-2P_f)^2}$. In the second case, $\llbracket H(j_+)=z\rrbracket +\llbracket H(j_-)=z\rrbracket$ is a Bernoulli variable with a success rate of $2P_f$, leading to a variance of $Var[\llbracket H(j_+)=z\rrbracket +\llbracket H(j_-)=z\rrbracket ]=(2P_f)(1-2P_f)$, and the variance of $\widehat{\underline{\mathbf{x}}}{j}$ is $Var[\widehat{\underline{\mathbf{x}}}{j}]=\frac{(2P_f)(1-2P_f)}{(P_t+P_o-2P_f)^2}$. In every $\mathbf{x}$, there are $s$ indices $j\in [d]$ satisfying the first case and $d-s$ indices satisfying the second case. Thus, the total error can be expressed as $\frac{s (P_t+P_o)(1-P_t-P_o)+(d-s)(2P_f)(1-2P_f)}{(P_t+P_o-2P_f)^2}$.

Consider Equation (\ref{eq:meanmse}) and three cases therein: $\llbracket j_+ \in \mathbf{Y}_\mathbf{x}\rrbracket =1$, $\llbracket j- \in \mathbf{Y}_\mathbf{x}\rrbracket =1$, $\llbracket j_b \in \mathbf{Y}_\mathbf{x}\rrbracket =0$ and $\llbracket j_b \in \mathbf{Y}_\mathbf{x}\rrbracket =0$. For the first case, the random variable $\llbracket H(j+)=z\rrbracket -\llbracket H(j_-)=z\rrbracket $ in Algorithm \ref{alg:cocoestimator} follows a probability distribution with the following probabilities:
\begin{equation*}
    \left\{
    \begin{array}{@{}lr@{}}
        \ \ \ 1,\ \ \ \ \text{with prob}\ \  P_t;\\
        \ \ \ 0,\ \ \ \ \text{with prob}\ \ 1-(e^\epsilon+1)/\Omega;\\
        -1,\ \ \ \ \text{with prob}\ \ P_o.\\
    \end{array}
    \right.
\end{equation*}
Therefore, the variance of the random variable $\llbracket H(j_+)=z\rrbracket -\llbracket H(j_-)=z\rrbracket $ is $Var[\llbracket H(j_+)=z\rrbracket -\llbracket H(j_-)=z\rrbracket ]=(P_t+P_o)-(P_t-P_o)^2$. Similarly, in the second case, the random variable $\llbracket H(j_+)=z\rrbracket -\llbracket H(j_-)=z\rrbracket $ follows a probability distribution with the following probabilities:
\begin{equation*}
    \left\{
    \begin{array}{@{}lr@{}}
        \ \ \ 1,\ \ \ \ \text{with prob}\ \  P_o;\\
        \ \ \ 0,\ \ \ \ \text{with prob}\ \ 1-(e^\epsilon+1)/\Omega;\\
        -1,\ \ \ \ \text{with prob}\ \ P_t.\\
    \end{array}
    \right.
\end{equation*}.
Thus, the variance of $\llbracket H(j_+)=z\rrbracket -\llbracket H(j_-)=z\rrbracket $ in this case is also given by $Var[\llbracket H(j_+)=z\rrbracket -\llbracket H(j_-)=z\rrbracket ]=(P_t+P_o)-(P_t-P_o)^2$. In the third case, the random variable $\llbracket H(j_+)=z\rrbracket -\llbracket H(j_-)=z\rrbracket $ follows a probability distribution with the following probabilities:
\begin{equation*}
    \left\{
    \begin{array}{@{}lr@{}}
        \ \ \ 1,\ \ \ \ \text{with prob}\ \  P_f;\\
        \ \ \ 0,\ \ \ \ \text{with prob}\ \ 1-2 P_f;\\
        -1,\ \ \ \ \text{with prob}\ \ P_f.\\
    \end{array}
    \right.
\end{equation*}

The variance of the difference between the indicator functions of $H(j_+)=z$ and $H(j_-)=z$ is equal to $2P_f$. Additionally, the variance of the estimator $\widehat{\overline{\mathbf{x}}}_{j}$ is $\frac{2P_f}{(P_t-P_o)^2}$. For each vector $\mathbf{x}$, there exist $s$ indices $j \in [d]$ that satisfy either the first or the second case, while $d-s$ indices satisfy the second case. Thus, the overall error can be expressed as $\frac{s((P_t+P_o)-(P_t-P_o)^2) +(d-s)(2P_f)}{(P_t-P_o)^2}$.

\section{Proof of Mean Squared Error Bounds of CoCo}\label{proof:cocomse}
For proving the Equation (\ref{eq:nmfmsebound}), we separately consider two formulas $\frac{(P_t+P_o)(1-P_t-P_o)}{(P_t+P_o-2P_f)^2}$ and $\frac{2P_f(1-2P_f)}{(P_t+P_o-2P_f)^2}$ in Equation (\ref{eq:nmfmse}). In the first formula, when $t=e^\epsilon s+5s$ and $\epsilon=O(1)$, we have $P_t+P_o=\frac{e^\epsilon+1}{s(e^\epsilon+1)+e^\epsilon s+4s}$ and $2P_f=\frac{2}{{s(e^\epsilon+1)+e^\epsilon s+4s}}$, thus $\frac{(P_t+P_o)(1-P_t-P_o)}{(P_t+P_o-2P_f)^2}\leq \frac{P_t+P_o}{(P_t+P_o-2P_f)^2} \leq \frac{c_1 s}{\epsilon^2}$ holds with some constant $c_1\in \mathbb{R}^+$ for any $s\in \mathbb{Z}^+$ and $0 < \epsilon=O(1)$. Similarly in the second formula,  $\frac{2P_f(1-2P_f)}{(P_t+P_o-2P_f)^2}\leq \frac{2P_f}{(P_t+P_o-2P_f)^2} \leq \frac{c_2 s}{\epsilon^2}$ holds with some constant $c_2\in \mathbb{R}^+$ for any $s\in \mathbb{Z}^+$ and $\epsilon=O(1)$. Therefore, the Equation (\ref{eq:nmfmse}) is now bounded by:
$$\frac{c_1 s+(d-s) c_2 s }{\epsilon^2}\leq O(\frac{d s}{\epsilon^2}).$$

For proving Equation (\ref{eq:meanmsebound}), we separately consider two formulas $\frac{(P_t+P_o)-(P_t-P_o)^2}{(P_t-P_o)^2}$ and $\frac{2P_f}{(P_t-P_o)^2}$ in Equation (\ref{eq:meanmse}). In the first formula, when $t=e^\epsilon s+s+2\geq 2s+2$ and $\epsilon=O(1)$, we have $P_{ow}\leq e^{-1}$, $P_t+P_o=\frac{e^\epsilon+1}{s(e^\epsilon+1)+e^\epsilon s+2}$ and $P_t-P_o=(1-P_{ow})   \frac{e^\epsilon-1}{s(e^\epsilon+1)+e^\epsilon s+2}$, thus  $\frac{(P_t+P_o)-(P_t-P_o)^2}{(P_t-P_o)^2}\leq \frac{P_t+P_o}{(P_t-P_o)^2} \leq \frac{c_3 s}{\epsilon^2}$ holds with some constant $c_3\in \mathbb{R}^+$ for any $s\in \mathbb{Z}^+$ and $\epsilon=O(1)$. Similarly in the second formula,  $\frac{2P_f}{(P_t-P_o)^2} \leq \frac{c_4 s}{\epsilon^2}$ holds with some constant $c_4\in \mathbb{R}^+$ for any $s\in \mathbb{Z}^+$ and $\epsilon=O(1)$. Therefore, the Equation (\ref{eq:meanmse}) is bounded by:
$$\frac{c_3 s+(d-s) c_4 s }{\epsilon^2}\leq O(\frac{d s}{\epsilon^2}).$$

\balance

\section{Proof of Mean Absolute Error Bounds of CoCo}\label{proof:cocomae}
Consider the $j$-th mean value $\overline{\mathbf{x}}_j$, according to Lemma \ref{lemma:cocomse}, the expectation of $\widehat{\overline{\mathbf{x}}}_j-\overline{\mathbf{x}}_j$ is $0$. Furthermore $\widehat{\overline{\mathbf{x}}}_j-\overline{\mathbf{x}}_j$ is the average of $n$ independent random variables $\frac{\llbracket H(j_+)=z]-[H(j_-)=z\rrbracket}{P_t-P_o}$, every of which lies in the range of:
 $$[\frac{-1}{P_t-p_o}, \frac{1}{P_t-P_o}].$$
When $\llbracket j_+ \in \mathbf{Y}_\mathbf{x}\rrbracket=1 \ or\  
 \llbracket j_- \in \mathbf{Y}_\mathbf{x}\rrbracket=1$, the variable has variation of $\frac{(P_t+P_o)-(P_t-P_o)^2}{(P_t-P_o)^2}$; when $\llbracket j_+ \in \mathbf{Y}_\mathbf{x}\rrbracket=0\ and\ \llbracket j_- \in \mathbf{Y}_\mathbf{x}\rrbracket=0$, the variable has variation of  $\frac{2P_f}{(P_t-P_o)^2}$. In both cases, since $P_{ow}\leq e^{-1}$ when $t\geq 2s+2$, we have $\frac{(P_t+P_o)-(P_t-P_o)^2}{(P_t-P_o)^2}\leq \frac{c}{s}$ and $\frac{2P_f}{(P_t-P_o)^2}\leq \frac{c_1 s}{\epsilon^2}$ for any $\epsilon=O(1), s\in \mathrm{R}^+$ with some constant $c_1\geq 0$.

According to the Bernstein inequalities on $n$ zero-mean bounded random variables \cite{uspensky1937introduction}, we have:
$$\mathbb{P}[|\widehat{\overline{\mathbf{x}}}_j-\overline{\mathbf{x}}_j|\geq \frac{a}{n}]\leq 2 \exp(-\frac{a^2/2}{n^2 Var[\widehat{\overline{\mathbf{x}}}_j-\overline{\mathbf{x}}_j]+a/(3 P_t-3 P_o)}).$$
Since $P_{ow}\leq e^{-1}$ when $t\geq 2s+2$, we have $1/(P_t-P_o)\leq \frac{c_2 s}{\epsilon}$ for any $\epsilon=O(1), s\in \mathbb{R}^+$ with some constant $c_1\geq 0$.

When $a\leq \frac{3 c_1 n}{c_2 \epsilon}$, we get $\mathbb{P}[|\widehat{\overline{\mathbf{x}}}_j-\overline{\mathbf{x}}_j|\geq \frac{a}{n}]\leq 2 \exp(-\frac{\epsilon^2 a^2/2}{2 c_1 n s}).$ Consequently, with probability $1-\beta$, we have $|\widehat{\overline{\mathbf{x}}}_j-\overline{\mathbf{x}}_j|\leq \sqrt\frac{2 c_1 s \log (2/\beta)}{\epsilon^2 n}$ (when $\sqrt\frac{2 c_1 s \log (2/\beta)}{\epsilon^2 n}\leq \frac{3 c_1 n}{c_2 \epsilon}$).

Now consider all $j \in [d]$, applying the union bound on tails probability $\frac{\beta}{d}$, we conclude that: with probability $1-\beta$, the $\max_{j=1}^d |\widehat{\overline{\mathbf{x}}}_j-\overline{\mathbf{x}}_j|\leq O( \sqrt\frac{s \log (d/\beta)}{\epsilon^2 n})$ holds (if $\sqrt\frac{2 c_1 s \log (2 d/\beta)}{\epsilon^2 n}\leq \frac{3 c_1 n}{c_2 \epsilon}$).